\documentclass[12pt]{article}
\usepackage[english]{babel}
\usepackage[left=2cm,right=2cm,top=2cm,bottom=2cm]{geometry}
\usepackage{siunitx}
\RequirePackage[OT1]{fontenc}
\RequirePackage{amsmath,amssymb,latexsym,subfigure} %amsthm,
\usepackage{epsfig,graphics}
\numberwithin{equation}{section}
\newcommand{\cqd}{\hfill\rule{2mm}{2mm}}
\newcommand{\car}{\mbox{\rm{1\hspace{-0.10 cm }I}}}

\newtheorem{teor}{Theorem}[section]

\newtheorem{rem}{Remark}[section]
\newtheorem{lema}{Lemma}[section]

\title{A Robbins-Monro algorithm for non-parametric estimation of  NAR process with Markov-Switching: asymptotic normality}

\author{\textbf{Lisandro J. Ferm\'in$^{1}$, Ricardo R\'{\i}os$^{2}$, Luis-\'Angel Rodr\'iguez$^{3}$} \\
     \small $^{1}$ {Aix-Marseille University, CNRS, Aix-Marseille School of Economics, France}\\
    \small $^{2}$ {Escuela de Matem\'aticas, Facultad de Ciencias, Universidad Central de 
                    Venezuela, Caracas, Venezuela}\\
    \small $^{3}$ {Departamento de Matem\'aticas, FACYT, Universidad de Carabobo, Valencia,
            Venezuela}}
\begin{document}

\maketitle

\begin{abstract} This paper is the second part of our study on the non-parametric estimation of MS-NAR processes started with \cite{Rico-Lisandro-Luis}. We consider the Nadaraya-Watson type regression function estimator for non-linear autoregressive Markov switching processes.
In this context the regression function estimator is interpreted as a solution of a local weighted 
%least-square problem which does not admit a closed-form  solution in the case of hidden Markov switching. 
We have introduced, in the first work, a restoration-estimation Robbins-Monro algorithm to approximate the estimator, and we  proved identifiability of model and the consistency of the non-parametric estimator. In this work, we obtain the central limit theorem for the non-parametric estimator, whether the Markov chain is observed or not. Finally, we present a detailed simulation study illustrating the performances of our estimation procedure.

% Please include a maximum of seven keywords
%\keywords{keyword 1, \emph{keyword 2}, keyword 3, keyword 4, keyword 5, keyword 6, keyword 7}
\end{abstract}

%%%%%%%%%%%%%%%%%%%%%%%%%%%%%%%%%%%%%%%%%%%%%%%%%%%%%%%%%%%%%%%%
%%%%%%%%%%%%%%%%%%%%%%%%%%%%%%%%%%%%%%%%%%%%%%%%%%%%%% Section 1

\section{Introduction}\label{sect.1}

\subsection{Overview of the literature}

In this paper, we are interested in nonparametric estimation for Markov Switching Non-linear Autoregressive (MS-NAR) processes. In this framework, the observation $Y_k$ at each time step is given by a noisy measurement of an unknown function of the previous observation $Y_{k-1}$. The unknown function is selected from a dictionary of functions identified by a label $X_k$ such that $\{X_k\}_{k\geq1}$ is a homogeneous Markov chain with state space $\{1,\ldots,m\}$.

Markov switching autoregressive processes can be viewed as a combination of hidden Markov models (HMMs) and threshold regression models. These switching autoregressive processes, introduced into the econometric context by Goldfeld and Quandt in \cite{Goldfeld}, have become quite popular in the literature. For instance, they have been employed by Hamilton \cite{Hamilton} in the analysis of the Gross Domestic Product (GDP) of the USA during both contraction and expansion regimes.

Parametric estimation theory in MS-NAR is well-developed; in particular, consistency and normality were proven in \cite{douc}. See \cite{cappe-moulines-ryden} for a recent state of the art in HMMs; this book presents efficient algorithms that allow for the computation of the likelihood function and the construction of practical inference methods. A new approach to demonstrate the consistency of the maximum likelihood estimator for MS-NAR was obtained in \cite{Marcano-Lisandro-Luis}, based on a uniform exponential memory loss property for the prediction filter (approximated by a filter with finite memory) and the $\alpha$-mixing property of the MS-NAR process.

In order to avoid misspecification that could arise from the use of parametric methods, research on nonparametric estimation methods for Markov models and HMMs has received significant attention. Identifiability of nonparametric mixtures of two multivariate distributions was investigated by Hall and Zhou \cite{Hall-Zhou} and extended by Allman \textit{et al.} \cite{Allman}. They showed generic identifiability of various latent-state models, including hidden Markov models with finite-valued observations. Their work relied on a seminal result in \cite{Kruskal} regarding the identification of factors in three-way tables. 

Nonparametric identification and maximum likelihood estimation for finite-state hidden Markov models were studied in \cite{Alexandrovich}. See also Gassiat and Rousseau \cite{gassiat4} for translation hidden Markov models with finite state spaces. They show that all model parameters (including infinite-dimensional parameters) are identifiable provided that the matrix defining the joint distribution of two consecutive latent variables is non-singular and the translation parameters are distinct. It is shown in \cite{gassiat5} that these fully nonparametric translation models are identifiable with respect to both the distribution of the latent variables and the distribution of the noise, generally under a light tail assumption on the latent variables. Identifiability in the context of nonparametric estimation of MS-NAR models is proven in \cite{Rico-Lisandro-Luis}. More recently, Balsells \textit{et al.} \cite{Balsells} established the identifiability of the latent variables and non-linear mappings in Switching Dynamical Systems up to affine transformations, by leveraging identifiability analysis techniques from identifiable deep latent variable models.

The use of semiparametric or non-parametric models to estimate mixtures of independent or Markov regime regressions has garnered recent interest, especially in econometric applications. For instance, \cite{Nademi} used a semiparametric Markov switching AR-ARCH model to forecast OPEC, WTI, and Brent crude oil prices. They assumed that the regression function factors into a non-linear component dependent on a parameter and a non-parametric adjustment. Their estimation was implemented via adaptations of the EM algorithm, supported by an exhaustive simulation study. A
semiparametric estimation for the MS-NAR model was studied in \cite{luis1}, where the authors
considered a conditional least-square approach for the parameter estimation and a kernel density estimator for the innovation density probability.

Similarly, in oceanography and meteorology, related models such as nonparametric State-space models (SSMs) have been applied to data assimilation (see \cite{Chau} and references therein). These authors proposed combining two techniques: the Stochastic Expectation–Maximization (SEM) algorithm and Sequential Monte Carlo (SMC) approaches, for non-parametric estimation in state–space models, accompanied by a simulation study. However, neither of these works addresses the theoretical properties of their respective algorithms. Therefore, we are interested in presenting a study that complements our work in \cite{Rico-Lisandro-Luis}.

\subsection{Contributions}

We consider  non-parametric regression estimators. That is, for $i=1,\ldots,m$, we define  a Nadaraya-Watson
type kernel estimator, given by
\begin{equation}\label{NWEstim}
\hat{r}_{i,n}(y)=\frac{\sum_{k=0}^{n-1}Y_{k+1}K\left(\frac{y-Y_k}{h}\right)\car_{i}(X_{k+1})}
{\sum_{k=0}^{n-1}K\left(\frac{y-Y_k}{h}\right)\car_{i}(X_{k+1})}.
\end{equation}
This Nadaraya-Watson type estimator was introduced for HMM models in \cite{Harel-Puri}. In our previous work, \cite{Rico-Lisandro-Luis},  we have studied the uniform consistency, assuming that a realization of the complete data $\{Y_0,(Y_k,X_k)\}_{k=1,\ldots, n}$ is known; i.e. 
\begin{equation*}
\sup_{y\in\mathrm{C}}|\hat{r}_{i,n}(y)-r_i(y)|\to 0, \, \, \mbox{a.s.} \,\,(\mbox{when } n\to\infty),
\end{equation*}
where  $\mathrm{C}$ is a compact subsets of $\mathbb{R}$.

Furthermore, when $X$ is a hidden Markov chain, the regression functions estimators $\hat{r}_{i,n}(y)$, for each $y$ and $i=1: m$, 
are interpreted as solutions $\theta=(\theta_1,\ldots,\theta_m)$ of the local weighted least-squares 
\begin{eqnarray}\label{potencialU}
U(y, Y_{0:n}, X_{1:n},\theta) = \frac{1}{nh}\sum_{k=0}^{n-1}\sum_{i=1}^{m}K\left(\frac{y-Y_k}{h}\right)\car_i(X_{k+1})(Y_{k+1}-\theta_i)^2,
\end{eqnarray}
where the weights are specified by the kernel $K$, so that the observations $Y_k$ near $y$ 
have the largest influence on the estimate of the regression function at $y$. That is,
\begin{equation*}
\hat{\theta}(y)=\underset{\theta\in\Theta\subset \mathbb{R}^m}{\mbox{argmin}}\; U(y, Y_{0:n}, X_{1:n},\theta).
\end{equation*} 
The solution of this problem must be approximated because it does not admit a closed-form.
We have proposed a recursive algorithm for the estimation of the regression
functions $r_{i}$ with a Monte-Carlo step which restores the missing data $X_{1:n}$
by $X_{1:n}^{t}$, and a Robbins-Monro procedure which allows us to estimate the unknown value of $\theta$.
This approximation minimizes the potential $ U $ using the gradient algorithm, for each fixed $y$ ,
$$
\theta^{t}=\theta^{t-1}-\gamma_t\nabla_{\theta} U\left(y, Y_{0:n}, X_{1:n}^t,\theta^{t-1}\right),
$$
where $\{\gamma_t\}$ is any sequence of real positive numbers decreasing to $0$, and $\nabla_{\theta} U$, the gradient of $U$ with respect to the vector $\theta\in \mathbb{R}^m$. The consistency of the estimator obtained by our Robbins-Monro algorithm is proved in 
 \cite{Rico-Lisandro-Luis}.

We establish in the present work the asymptotic normality of the Naradaya-Watson kernel estimator in the case of complete data, 
as well as the asymptotic normality of the estimator related to our Robbins-Monro algorithm.

The paper is organized as follows. In Section \ref{sect.2}, we pres\-ent the preliminary results. We give the general conditions on the model ensuring
the existence of a probability density distribution and the model stability. Moreover, sufficient conditions are given to ascertain the strong mixing dependence and the model identifiability. 
In Section \ref{sect.3}, we prove the asymptotic normality for the non-parametric estimators in the case of complete data as well 
as in the case of partially observed data. 
We recall the uniform consistency of the Naradaya-Watson kernels estimator in the case of complete data, and the consistency of estimator related to our Robbins-Monro algorithm. 
Section \ref{sect.4} contains some numerical experiments on simulated data illustrating the 
performance of our non-parametric estimation procedure. The proofs are deferred to the Appendix \ref{appA}.

%%%%%%%%%%%%%%%%%%%%%%%%%%%%%%%%%%%%%%%%%%%%%%%%%%%%%%%%%%%%%%%%
%%%%%%%%%%%%%%%%%%%%%%%%%%%%%%%%%%%%%%%%%%%%%%%%%%%%%% Section 2

\section{Preliminary}\label{sect.2}

We consider a particular type of switching non-linear autoregressive models with Markov regime, called Markov Switching Non-linear Autoregressive (MS-NAR)
processes and which is defined for $k\geq 1$ by
\begin{equation}\label{modelo}
    Y_k=r_{X_k}(Y_{k-1})+e_k
\end{equation}
where $\{e_k\}_{k\geq1}$ is a sequence of i.i.d. random variables,  the sequence $\{X_k\}_{k\geq1}$
is an homogeneous Markov chain with state space $\{1,\ldots,m\}$, and $r_1(y),\ldots, r_m(y)$ are the regression functions, assumed
to be unknown.

We denote by $A$ the probability transition matrix of the Markov chain $X=\{X_k\}_{k\geq1}$, i.e. $A=(a_{ij})_{i,j=1:m}$, with $a_{ij}=\mathbb{P}(X_k=j|X_{k-1}=i)$. 
We assume that the variable $Y_0$, the Markov
chain $X$ and the sequence $\{e_k\}_{k\geq1}$ are mutually independent.

This model is a generalization of the switching linear autoregressive model with Markov regime, also known as MS-AR model. When the regression functions $r_i$ are linear the MS-NAR process is simply an MS-AR model. 

We are obliged to recall some features on the MS-NAR model. 

%%%%%%%%%%%%%%%%%%%%%%%%%%%%%%%%%%%%%%%%%%%%%%%%%%%%%% Section 2.1 

\subsection{Stability and existence of moments}

This section includes known
results given by \cite{Yao}. Our aim is to summarize the sufficient conditions 
which ensure the existence and uniqueness of a strictly stationary ergodic solution for the model, as well as the existence for the respective stationary distribution of  moments of order $s\geq 1$. 

\begin{enumerate}
\item[\bf E1] The Markov chain $\{X_k\}_{k\geq 1}$ is positive recurrent with probability transition matrix $A=(a_{ij})_{i,j=1:m}$. Hence, it has an invariant distribution $\mu=(\mu_1,\ldots,\mu_m)$.

\item[\bf E2] The functions $r_i$, for $i=1,...,m$, are continuous. 

\item[\bf E3] There exist positive constants $\rho_i, b_i$, $i=1,...,m$,  such that $|r_i(y)|\leq \rho_i|y|+b_i$, for $y\in\mathbb{R}$.

\item[\bf E4] $\sum_{i=1}^m \mu_i \log \rho_{i}<0$. 

\item[\bf E5] $\mathbb{E}(|e_1|^s)<\infty$, for some $s\geq 1$.

\item[\bf E6] The sequence $\{e_k\}_{k\geq1}$  of random variables has a common density probability function $\Phi(e)$ 
with respect to the Lebesgue measure.

\item[\bf E7] The density probability function $\Phi(e)$  is everywhere positive on $\mathbb{R}$.
\end{enumerate}

Conditions  E1-E7 ensure the existence and uniqueness of a strictly stationary ergodic solution for the model, see results given by Yao and Attali (1999 \cite{Yao}. In addition they also consider hypothesis for the existence of moments; for instance, if the spectral radius of $Q_s=\left( \rho_j^s \, a_{ij}\right)_{i,j=1\ldots m}$ is strictly 
less than 1, with $s$ the same as in condition E5, then $\mathbb{E}(|Y_k|^s)<\infty$. They argue the necessity and sufficiency of all the hypotheses considered in their results. 

%%%%%%%%%%%%%%%%%%%%%%%%%%%%%%%%%%%%%%%%%%%%%%%%%%%%%%  Section 2.2 

\subsection{Probability density}

We present a technical result which states the existence of conditional densities of the MS-NAR model and gives a factorization of this density probability. This result is relevant in the frame of kernel estimation.

Let us first introduce some notations: $V_{1:n}$ stands for the random vector $(V_1,\ldots,V_n)$, and $v_{1:n}=(v_1,\ldots,v_n)$ is
a realization of the respective random vector, and $p(V_{1:n}=v_{1:n})$ denotes the density distribution of random vector $V_{1:n}$ evaluated at $v_{1:n}$. 

%The symbol $\car_B(x)$ denotes the indicator function of set $B$, which assigns the value $1$ if $x\in B$ and  $0$ otherwise.

We consider the following assumption:

\begin{enumerate}
\item[\bf D1] The random variable $Y_0$ has a density function with respect to Lebesgue measure, denoted by $p(Y_0=y_0)$.
\end{enumerate}

In \cite{Rico-Lisandro-Luis} we prove that the random vector $(Y_{0:n},X_{1:n})$ admits the joint probability density 
\begin{eqnarray}
\label{exisdensidad_i}
\lefteqn{p(Y_{0:n}=y_{0:n},X_{1:n}x_{1:n})}\\
\nonumber 
&=&\Phi(y_n-r_{x_n}(y_{n-1}))a_{x_{n-1}x_n}\; \cdots\; \Phi(y_2-r_{x_2}(y_{1}))a_{x_{1}x_2}
\Phi(y_1-r_{x_1}(y_{0}))\mu_{x_1}p(Y_0=y_0),
\end{eqnarray}
with respect to the product measure $\lambda\otimes\mu_c$, where $\lambda$ and $\mu_c$ denote Lebesgue and counting measures, respectively. Moreover, if $\Phi$ is a bounded density, then the joint density of $(Y_k,Y_{k'})$
satisfies 
\begin{equation}\label{exisdensidad_ii}
p(Y_k=y_k,Y_{k'}=y_{k'})\leq \|\Phi\|_{\infty}^2\, \mbox{ for }\,  k,k'\geq 1,\, \mbox{and}\,
p(Y_0=y_0,Y_{k'}=y_{k'})\leq \|\Phi\|_{\infty}\,\mbox{ for }\, k'\geq 1.
\end{equation}

%%%%%%%%%%%%%%% Lemma 2.1 
%\begin{lema}\label{exisdensidad}[Ferm\'in, R\'ios and Rodr\'iguez (2017)]
%Under conditions D1 and E6,
%\begin{itemize}
%\item[i)] The random vector $(Y_{0:n},X_{1:n})$ admits the joint probability density 
%\begin{eqnarray*}
%p(Y_{0:n}=y_{0:n},X_{1:n}=x_{1:n})&= &\Phi(y_n-r_{x_n}(y_{n-1}))a_{x_{n-1}x_n}\; \cdots\; \Phi(y_2-r_{x_2}(y_{1}))a_{x_{1}x_2}
%\Phi(y_1-r_{x_1}(y_{0}))\mu_{x_1}p(Y_0=y_0),
%\end{eqnarray*}
%with respect to the product measure $\lambda\otimes\mu_c$, where $\lambda$ and $\mu_c$
%denote Lebesgue and counting measures, respectively. 
%\item[ii)] If $\Phi$ is a bounded density, then the joint density of $(Y_k,Y_{k'})$
%satisfies 
%\begin{equation*}
%p(Y_k=y_k,Y_{k'}=y_{k'})\leq \|\Phi\|_{\infty}^2\ \quad \mbox{ for } \quad k,k'\geq 1, \quad \mbox{and} \quad
%p(Y_0=y_0,Y_{k'}=y_{k'})\leq \|\Phi\|_{\infty}\ \quad \mbox{ for } \quad k'\geq 1.
%\end{equation*}
%\end{itemize}
%\end{lema}

%%%%%%%%%%%%%%%%%%%%%%%%%%%%%%%%%%%%%%%%%%%%%%%%%%%%%% Section 2.3 

\subsection{Strong mixing}\label{strongmixing}

A strictly stationary stochastic process $Y =\{Y_k\}_{k\in\mathbb{Z}}$ is strong mixing, if 
\begin{equation*}
\alpha_n=\sup\{|\mathbb{P}(A\cap B)-\mathbb{P}(A)\mathbb{P}(B)|: A\in \mathcal{M}_{-\infty}^0, B\in \mathcal{M}_n^{\infty}\} \to 0,
\end{equation*}
as $n\to \infty$, where $\mathcal{M}_a^{b}$, with $a,b \in \overline{\mathbb{Z}} $, is the $\sigma$-algebra gen\-e\-rated by $\{Y_k\}_{k= a:b}$. The values $\alpha_n$ are called strong mixing coefficients. For properties under mixing assumptions see \cite{dmixing}.

 In \cite{Rico-Lisandro-Luis} we prove that the MS-NAR model under conditions E1-E7  is $\alpha$-mixing and their coefficients $\alpha_n(Y)$ decrease geometrically.

%%%%%%%%%%%%%%% Proposition 2.2
%\begin{prop} \label{ModeloesMixing}[Ferm\'in, R\'ios and Rodr\'iguez (2017)]
%The MS-NAR model under conditions E1-E7  is $\alpha$-mixing and their coefficients $\alpha_n(Y)$ decrease geometrically.
%\end{prop}

%%%%%%%%%%%%%%%%%%%%%%%%%%%%%%%%%%%%%%%%%%%%%%%%%%%%%%  Section 2.4 

\subsection{Identifiability}

We assume the following conditions:

\begin{enumerate}
\item[\bf I1] The probability transition matrix $A=(a_{ij})_{i,j=1:m}$ has full rank.
\item[\bf I2] The functions $r_1,\!\ldots,r_m$ are different a.s.; i.e. if $i\neq j$ then $r_i(y')\neq r_j(y')$ for almost all $y'$.
\item[\bf I3] The density function $\Phi$ is such that $\{\Phi(y-r_i(y'))\}_{i=1:m}$ are linearly independents; i.e. 
$\sum_{i=1}^m\!\alpha_i \Phi(y-r_i(y'))=0$ $\forall y,y'$ if and only if  $\alpha_1=\ldots=\alpha_m=0$.
\item[\bf I4] The density function $\Phi$ is such that 
$\Phi(y-\tilde{r}_{\tilde{k}}(y'))= \Phi(y- r_k(y'))$ $\forall y$
if and only if $\tilde{r}_{\tilde{k}}(y')=r_k(y')$.
\end{enumerate}

Under conditions I1-I4 and  assuming that the number of states  $m$ is known,  in \cite{Rico-Lisandro-Luis} we prove that $A$ and $r$ are identifiable up to label swapping of the hidden states.

%%%%%%%%%%%%%%% Proposition 2.3 
%\begin{prop}\label{identifiability}[Ferm\'in, R\'ios and Rodr\'iguez (2017)] Assume that the number of states  $m$ is known, under conditions I1-I4, $A$ and $r$ are identifiable up to label swapping of the hidden states.
%\end{prop}

%%%%%%%%%%%%%% Remark 2.2
\begin{rem}\label{remark.2.2}
Condition I2 implies the identifiability of regression functions $r_i$'s for almost all $y'$. Nevertheless, the continuity given by condition E2 ensures the identifiability for all $y'$.

In the case where the innovation $e$ is a Gaussian white noise, conditions I1 and I2 are sufficient to obtain the identifiability.
\end{rem}

%%%%%%%%%%%%%%%%%%%%%%%%%%%%%%%%%%%%%%%%%%%%%%%%%%%%%%%%%%%%%%%%
%%%%%%%%%%%%%%%%%%%%%%%%%%%%%%%%%%%%%%%%%%%%%%%%%%%%%% Section 3

\section{Main results}\label{sect.3}

In this section we establish the asymptotic normality of the Naradaya-Watson kernel estimator in the case of complete data, as well as the asymptotic normality of estimator related to our Robbins-Monro algorithm.

%%%%%%%%%%%%%%%%%%%%%%%%%%%%%%%%%%%%%%%%%%%%%%%%%%%%%% Section 3.1 
\subsection{Assumption of the estimators}

First we assume that a realization of the complete data  $Y_{0:n}, X_{1:n}$ 
is available. We give a Central Limit Theorem
for the non-parametric estimator $\hat{r}_{i,n}$ of the regression function $r_i$, in each regime $i$. 

For a stationary MS-NAR model, $r(y)=\mathbb{E}(Y_1|Y_0=y)$ is the quantity of interest in the regression function estimation. It can be rewritten as
\begin{equation*}
r(y)=\sum_{i=1}^m \mathbb{E}(Y_1|Y_0=y,X_1=i)\mathbb{P}(X_1=i).
\end{equation*}
Hence, it is sufficient to estimate each regression function 
\begin{equation*}\label{regfunction}
r_i(y)=\mathbb{E}(Y_1|Y_0=y,X_1=i), \quad i=1,\ldots,m \; \mbox{ and } \; y\in\mathbb{R}.
\end{equation*}
Let us introduce
\begin{eqnarray*}
g_i(y)&:=&r_i(y)f_i(y),\\
f_i(y)&:=&\mu_ip(Y_0=y).
\end{eqnarray*}
The Nadaraya-Watson kernel estimator of $r_i$ is  
\begin{equation*}
\hat{r}_{i,n}(y)=\left\{\begin{array}{cc}
{\hat{g}_{i,n}(y)}/{\hat{f}_{i,n}(y)} & \mbox{if $\hat{f}_{i,n}(y)\not=0$,}\\
0& \mbox{otherwise}
\end{array}\right.
\end{equation*}
with
\begin{eqnarray*}
\hat{g}_{i,n}(y)&:=&\frac{1}{nh}\sum_{k=0}^{n-1}Y_{k+1}K_h\left({y-Y_k}\right)\car_{i}(X_{k+1}),\\
\hat{f}_{i,n}(y)&:=&\frac{1}{nh}\sum_{k=0}^{n-1}K_h\left({y-Y_k}\right)\car_{i}(X_{k+1}),
\end{eqnarray*}
and $K_h(y)= K(y/h)$.

Let us take a kernel $K:\mathbb{R}\to\mathbb{R}$, positive, symmetric, with compact support such that $\int K(t)dt=1$. 
We assume that the kernel $K$ as well as  the density $\Phi$ are bounded, i.e  
\begin{enumerate}
\item[\bf B1] $\|K\|_{\infty}<\infty$.
\item[\bf B2] $\|\Phi\|_{\infty}<\infty$.
\end{enumerate}

Under condition B1, the kernel $K$ is of order 2, i.e. $\int tK(t)dt=0$  and $0<\left|\int t^{2}K(t)dt\right|<\infty$. 

Let $\mathrm{C}$ be a compact subset of $\mathbb{R}$. We assume the following regularity conditions: 
 
\begin{enumerate}
\item[\bf R1] There exist finite constants $c, \beta>0$, such that
$\forall y,y'\in \mathrm{C},\ |K(y)-K(y')|
<c|y-y'|^\beta$.
\item[\bf R2] The density of $Y_0$, $\Phi$, and $r_i$ have continuous second derivatives in the interior of $\mathrm{C}$.
\item[\bf R3] For all $k\in\mathbb{N}$, the functions
$r_{i,k}(t,s)=\mathbb{E}(|Y_1Y_{k+1}|\, |Y_0=t,Y_k=s,X_1=i,X_{k+1}=i)$
are continuous and uniformly bounded with respect to $k$. 
\end{enumerate}

Finally, we impose one of the two following moment conditions:
\begin{enumerate}     
\item[\bf M1] $\mathbb{E}(\exp(|Y_0|))<\infty$ and $\mathbb{E}(\exp(|e_1|))<\infty$.
\item[\bf M2] $\mathbb{E}(|Y_0|^s)<\infty$ and $\mathbb{E}(|e_1|^s)<\infty$,  for some $s>2$.
\end{enumerate}

%%%%%%%%%%%%%%% Remark 2.3
\begin{rem}\label{remark.2.3} 
Condition M1 and M2 are assumed to hold. The former in order to obtain the a.s. uniform convergence over compact sets and the latter for the a.s. pointwise convergence.
\end{rem}

Let  $\{h_n\}_{n\geq 1}$ be a sequence of real numbers satisfying the following condition  
\begin{enumerate}
\item[\bf S1] For all $n\geq0$, $h_n>0$, $\lim_{n\to\infty}h_n=0$ and $\lim_{n\to\infty}nh_n=\infty$. 
\end{enumerate}

In \cite{Rico-Lisandro-Luis}, we have established the uniform consistency over compact sets of the Nadaraya-Watson kernel estimator $\hat{r}_{i,n}$ defined in \eqref{NWEstim}. Urder conditions E1-E4, E6-E7, D1, B1-B2, S1, and R1-R3 on a compact set $C$, we have: 
\begin{enumerate}
\item[i)] If $nh_n/\log n \to \infty$ and condition M2 holds, then $|\hat{r}_{i,n}(y)-r_i(y)|\to 0$ a.s. for all $y\in C$.
\item[ii)] If $nh_n/\log n \to \infty$ and condition M1 holds, then $\sup_{y\in\mathrm{C}}|\hat{r}_{i,n}(y)-r_i(y)|\to 0$ a.s.
\end{enumerate}

Here, we define 
\begin{equation*}
{g_2}_{i}(y)=f_i(y)\mathbb{E}(Y_1^2|Y_0=y,X_1=i).
\end{equation*}
We denote by $A^{(k)}_{ij}$ the $(i,j)$-th entry
of the $k$-th power of the matrix $A$. We define $B_{n,h}\approx B_h$ so that $\lim_{h\to 0}\lim_{n\to \infty} B_{n,h}$ $= \lim_{h\to 0}B_h$, i.e. for large enough $n$ and small enough  $h$, $B_{n,h}$ is approximately equal to $B_h$. Analogously, we define $B_{n,h} \preceq B_h$ to mean that $\lim_{h\to 0} \lim_{n\to \infty} B_{n,h} \leq \lim_{h\to 0} B_h$. In particular we write
$B_{n,h} \preceq B$ to mean that $B$ is a bound for the sequence $B_{n,h}$, for large enough $n$ and small enough $h$ . 

\subsection{Asymptotic normality: fully observed data case}

%%%%%%%%%%%%%%% Theorem 3.1
\begin{teor}
\label{normalidadTO}
Assume that the model MS-NAR satisfies the conditions E1-E4, E6-E7, D1, B1-B2, M2, S1, and R2-R3 on a compact set $C$. Then, for each fixed $y\in C$, 
\begin{equation*}
\sqrt{nh_n}(\hat{r}_{i,n}(y)-r_i(y))\to\mathcal{N}\left(0,\frac{\sigma_i^2\|K\|_2^2}{f_i(y)}\right),
\end{equation*} 
whenever $nh_n\to \infty$.
\end{teor}

The proof is reported to appendix A. 

We present the following technical lemmas that allows to apply the truncated device in order to obtain the CLT for the Nadaraya-Watson kernel estimator. 

%%%%%%%%%%%%%%% Lemma 3.1 
\begin{lema}\label{asympt_cov} [Ferm\'in, R\'ios and Rodr\'iguez (2017)] Assume that the model MS-NAR satisfies the conditions E1-E7, D1, B1-B2, S1, and R2-R3  on a compact set $C$. Let $\{M_n\}_{n\geq 1}$ be a non-decreasing sequence of positive numbers tending to infinity. Let 
\begin{equation*}
T_{k,n}= aK_h(y-Y_k)\car_i(X_{k+1})+bY_{k+1}\car_{\{|Y_{k+1}|\leq M_n\}}K_h(y-Y_k)\car_i(X_{k+1}).
\end{equation*}
Then, the following statements hold, for all $y\in C$:
\begin{itemize}
\item[i)]  $\mathrm{Var}(T_{0,n})\!\approx\! h(\!a^2f_i(y)\!+\!2abg_i(y)\!+\!b^2g_{2,i}(y))\|K\|_2^2+o(h^2)$.
\item[ii)] $\mathrm{Cov}(T_{0,n},T_{k,n})\preceq h^2\left(a^2+2ab(|r_i(y)|+\mathbb{E}(|e_1|))\right) +h^2b^2 r_{i,k}(y,y))A^{(k)}_{i,i}\mu_i\|\Phi\|_{\infty} +o(h^3)$.
\item[iii)] $\mathrm{Cov}(T_{0,n},T_{k,n})\leq (a^2 + 2ab M_n+b^2M_n^2)4\|K\|_{\infty}^2 \alpha_{k}$, for any $n>0$.
\end{itemize}
\end{lema}

Lemma~\ref{asympt_cov} allows treating in a unified way the asymptotic behavior of the variances and covariances of $\hat{f}_{i,n}$ and a truncated version of $\hat{g}_{i,n}$. In contrast, the asymptotic behavior of the variances and covariances of the remaining part of the truncation of $\hat{g}_{i,n}$ is treated in Lemma \ref{asympt_cov2}.

%%%%%%%%%%%%%%% Lemma 3.2 
\begin{lema}\label{asympt_cov2} Assume that the model MS-NAR satisfies the conditions E1, E3, E6, D1, B1-B2, S1, M2 and R3  on a compact set $C$. Let $\{M_n\}_{n\geq 1}$ be a non-decreasing sequence of positive numbers tending to infinity. Let 
$$
R_{k,n}=b Y_{k+1}\car_{\{|Y_{k+1}|> M_n\}}K_h\left(y-Y_k\right)\car_i(X_{k+1}).
$$
Then, the following statements hold, for all $y\in C$:
\begin{itemize}
\item[i)]  $\mathrm{Var}(R_{0,n})\leq b^{2}\mathbb{E}(|Y_1|^{s})\|K \|_{\infty}^2M_n^{2-s}$, for any $n>0$.
\item[ii)] $\mathrm{Cov}(R_{0,n},R_{k,n})\leq b^{2}\mathbb{E}(|Y_1|^{s})\|K \|_{\infty}^2M_n^{2-s}$, for any $n>0$.
\item[iii)] $\mathrm{Cov}(R_{0,n},R_{k,n})\preceq  h^2b^2r_{i,k}(y,y)\|\Phi\|_{\infty} + o(h^3)$.
\end{itemize}
\end{lema}

The proof is reported to appendix A. 

%%%%%%%%%%%%%%%%%%%%%%%%%%%%%%%%%%%%%%%%%%%%%%%%%%%%%% Section 3.2 

\subsection{Asymptotic normality: partially observed data case}

In this section we present our Robbins-Monro type algorithm for the non-parametric estimation of the MS-NAR process in the partially observed data case. We compute the rate of convergence of our Robbins-Monro algorithm assuming that the almost surely convergence of the sequence $\{{\theta}^t(y)\}_{t\geq 0}$ was obtained for each fixed $y$. Then, we obtain the asymptotic normality for $\bar{\theta}^t(y)$.

The Nadaraya-Watson estimator $\hat{r}_n=(\hat{r}_{1,n},\ldots, \hat{r}_{m,n})$ can be interpreted, for each fixed $y$, as the solution of a locally weighted least-squares problem; in our case this has to do with finding the minimum of the potential $U$, defined by \eqref{potencialU},
with respect to $\theta=(\theta_1,\ldots,\theta_m)$ in a convex open set $\Theta$ of $\mathbb{R}^{m}$. 
Thus, the regression estimator $\hat{r}_n$ is given by 
\begin{equation*}
\hat{r}_n(y)=\underset{\theta\in\Theta \subset \mathbb{R}^m}{\mbox{argmin}}\;\;
U(y,Y_{0:n},X_{1:n},\theta).
\end{equation*}

In the partially observed data case, that is, when we do not observe $\{X_k\}_{k\geq 1}$, we cannot obtain an explicit expression for the solution $\hat{r}_{n}(y)$. Thus, we must consider  
a recursive algorithm for the approximation this solution. Our approach approximates the estimator $\hat{r}_n(y)$ by a
stochastic recursive algorithm similar to that of Robbins-Monro, \cite{cappe-moulines-ryden,duflo,YaoRe}.
This involves two steps: first, a Monte-Carlo step which restores the
missing data $X_{1:n}$, and second, a Robbins-Monro approximation in order to minimize the potential $U$.

At this point we introduce some further notation. For each $1\leq i\leq m$, $n_i(X_{1:n})=\sum_{k=1}^n\car_{i}(X_k)$ is the number of visits of the Markov chain $X$ 
to state $i$ in the first $n$ steps, and $n_{ij}(X_{1:n})=\sum_{k=1}^{n-1}\car_{i,j}(X_{k-1},X_k)$ is the number of transitions from $i$ 
to $j$ in the first $n$ steps. $\psi^{t}=(\theta^t,A^t)$ is a vector containing the estimated functions $\theta^t=(\theta_1^t,\ldots, \theta_m^t)$ and the 
estimated probability transition matrix $A^t$, in the $t$-th iteration of the Robbins-Monro algorithm.

The restoration-estimation Robbins-Monro algorithm consists  of the following steps, for each fixed $y$:
\begin{description}
\item[\textbf{Step 0.}] Pick an arbitrary initial realization  $X_{1:n}^{0}=X_1^{0},\ldots,X_n^{0}$. Compute 
the estimated regression functions $\hat{r}_n^{0}(y)=(\hat{r}_{1,n}^{0}(y),\ldots,\hat{r}_{m,n}^{0}(y))$
from equation \eqref{NWEstim} in terms of the observed data $Y_{0:n}$ and the initial realization $X_{1:n}^{0}$.
We compute the estimated transition probability matrix 
$A^{0}=(a^{0}_{ij})_{i,j=1:m}$, by  $a_{ij}^{0}=n_{ij}(X_{1:n}^{0})/n_i(X_{1:n}^{0})$ for $i,j=1,\ldots,m$
and the initial probability measure $\lambda^0=(\lambda^{0}_{1:m})$ by
$\lambda^{0}_i=n_i(X_{1:n}^{0})/n$ for $i=1,\ldots,m$. Then, define $\theta^{0}=\hat{r}^{0}_n(y)$.
\end{description}
For $t\geq1$,
\begin{description}
\item[\textbf{Step R.}] Restore the corresponding unobserved 
data by drawing a sample $X_{1:n}^{t}$ from the joint conditional distribution $p(X_{1:n}|Y_{0:n},\psi^{t-1})$.

\item[\textbf{Step E.}]  Update the estimation $\psi^{t}=(\theta^t,A^t)$ by
$$\theta^{t}=\theta^{t-1}-\gamma_t\nabla_\theta U\left(y,Y_{0:n},X_{1:n}^{t},\theta^{t-1}\right),$$
with $\nabla_\theta U\left(y,Y_{0:n},X_{1:n}^{t},\theta^{t-1}\right)=\nabla_\theta U\left(y,Y_{0:n},X_{1:n}^{t},\theta\right)
\bigm|_{\theta=\theta^{t-1}}$, the transition probability matrix  $A^{t}=(a^{t}_{ij})_{i,j=1:m}$ by
$a^{t}_{ij}=n_{ij}(X_{1:n}^{t})/n_i(X_{1:n}^{t})$ for $i,j=1,\ldots,m$,
and the probability measure $\lambda^t=(\lambda^{t}_i)_{i=1:m}$ by $\lambda^{t}_i=n_i(X_{1:n}^{t})/n$ for $i=1,\ldots,m$.

\item[\textbf{Step A.}] Reduce the asymptotic variance of the algorithm by
using the averages $\bar{\theta}^{t}=\frac{1}{t}\sum_{k=1}^t\theta^k$ instead of $\theta^{t}$,
which can be recursively computed taking $\bar{\theta}^{0}=\theta^0$, and 
$$\bar{\theta}^{t}=\bar{\theta}^{t-1}+\frac{1}{t}\left(\theta^{t}-\bar{\theta}^{t-1}\right).$$
\end{description}

%%%%%%%%%%%%%%% Theorem 3.2 
\begin{teor}
\label{tasadeconvergencia}
Assume that the closure of $\{\theta^t\}_{t\geq 0}$ is a compact subset of $\Theta$, $\theta^*=\lim_{t\to\infty}\theta^t$ a.s, and $\{\gamma_t\}$ is a positive sequence such that $\gamma_t\to 0$, $\frac{\gamma_t - \gamma_{t+1}}{\gamma_t}=o(\gamma_t)$, when $t\to \infty$. Then, for each fixed $y$,
\begin{equation*}
{\sqrt{t}}(\bar{\theta}^{t}(y)-\theta^*(y))\to\mathcal{N}\left(0,\Sigma^*_n(y)\right), \quad \mbox{ as } \quad  t\to \infty,
\end{equation*}
where $\Sigma^*_n(y)=\nabla_\theta^2 u(y,Y_{0:n},\theta^*)^{-1}\Gamma \nabla_\theta^2 u(y,Y_{0:n},\theta^*)^{-1}$.
\end{teor}

The proof is reported to Appendix \ref{appA}.

%%%%%%%%%%%%%%% Remark 2.4
\begin{rem}\label{invmeasure}
The invariant probability distribution of the  ergodic Markov chain $\{X_{1:n}^{t}\}_{t\in\mathbb{N}}$ 
is $p(X_{1:n}=x_{1:n}|Y_{0:n},\psi^*)$, with $\psi^*=\lim_{t\to \infty} \psi^t$.
\end{rem}

Let,
$
\mathbb{E}_{\psi'}\left(U(y,Y_{0:n},X_{1:n}^{t},\theta)|\mathcal{F}_{t-1} \right)=u(y,Y_{0:n},\theta)
$
where  $\mathbb{E}_{\psi'}(\cdot)=\mathbb{E}(\cdot |Y_{0:n}, {\psi'})$ and $\psi'=(\theta',A') \in \mathcal{F}_{t-1}$ the $\sigma$-algebra generated by $\{X_{1:n}^s\}_{s=0:(t-1)}$. This conditional expectation is in fact the expectation with respect to the conditional distribution function $p(X^t_{1:n}|Y_{0:n},\psi')$. Then,  for each $\theta\in\Theta$
\begin{equation*}\label{constrastelimite}
u(y,Y_{0:n},\theta) =\frac{1}{nh}\sum_{k=0}^{n-1}\sum_{i=1}^{m}K_h(y-Y_k)p(X_{k+1}=i|Y_{0:n},\psi')(Y_{k+1}-\theta_i)^2, \nonumber
\end{equation*}
and
$
\mathbb{E}_{\psi'}\left(\nabla_\theta U(y,Y_{0:n},X_{1:n}^{t},\theta)|\mathcal{F}_{t-1}\right)=\nabla_\theta u(y,Y_{0:n},\theta).
$

This implies that the Restoration-Estimation  algorithm is a stochastic gradient algorithm
that minimizes $u(y,Y_{0:n},\theta)$ and can be written as
\begin{equation}
\label{GradienteEst}
\theta^{t}=\theta^{t-1}+\gamma_t\left(-\nabla_\theta u(y,Y_{0:n},\theta^{t-1})+\varsigma_t\right),
\end{equation}
where
\begin{equation*}
\varsigma_t = -\nabla_\theta U\left(y,Y_{0:n},X_{1:n}^{t},\theta^{t-1}\right)+\nabla_\theta u(y,Y_{0:n},\theta^{t-1}). 
\end{equation*}
Thus, the stochastic gradient algorithm is obtained by perturbation of the following gradient system
\begin{equation*}
\dot{\theta}=-\nabla_\theta u(y,Y_{0:n},\theta).
\end{equation*}

The following convergence result is obtained for each fixed $y$. 

In  \cite{Rico-Lisandro-Luis}, assuming condition B1, that the closure of the set $\{\bar{\theta}^{t}\}$ is a compact subset of
$\Theta$, and that $\{\gamma_t\}$ is a positive sequence such that
$\sum_t \gamma_{t}=\infty$, $\sum_t\gamma_{t}^2<\infty$, we prove that the sequence  $\{\bar{\theta}^{t}\}$ 
satisfies 
$\lim_{t\to\infty}\nabla_{\theta}u(y,Y_{0:n},\bar{\theta}^t)=0$.
Furthermore, 
$\lim_{t\to\infty}{\bar{\theta}}^t=\theta^*$ and $\nabla_{\theta}u(y,Y_{0:n},\theta^*)=0$, a.s.

%%%%%%%%%%%%%%% Remark 2.5
\begin{rem}\label{ThetaAster}
The $i$-th component $\theta_i^{*}$ of the critical point $\theta^{*}\in \Theta\subset \mathbb{R}^m$ of the gradient of $u$, is given by
\begin{equation*}
\theta_i^{*}(y,Y_{0:n})=\frac{\sum_{k=0}^{n-1}Y_{k+1}K_h(y-Y_k)p(X_{k+1}=i|Y_{0:n},\psi^{*})}
{\sum_{k=0}^{n-1}K_h(y-Y_k)p(X_{k+1}=i|Y_{0:n},\psi^{*})}
= \frac{\mathbb{E}\left[\hat{g}_{i,n}(y)|Y_{0:n},\psi^{*}\right]}{\mathbb{E}\left[\hat{f}_{i,n}(y)|Y_{0:n},\psi^{*}\right]},
\end{equation*}
where 
$\psi^{*}=(\theta^{*}, A^{*})$, $\theta^{*}=\lim_{t\to\infty} \bar{\theta}^{t}$,    
and $A^{*}=\lim_{t \to \infty} A^t$ is the probability transition of the limit Markov chain $X$
and $\lambda^*=(\lambda^{*}_i)_{i=1:m}$ with $\lambda^{*}=\lim_{t \to \infty} \lambda^t$.

Since for all $t\geq 0$ we have $\lambda^t A^t=\lambda^t$, we deduce that $\lambda^{*} A^{*}= \lambda^{*}$. Now, taking $n\to \infty$, it is easy to verify that $A^{*}(Y_{0:n}) \to A$ and $\lambda^*(Y_{0:n}) \to \mu$, where $A$ is the probability transition matrix and $\mu$ an invariant measure of the Markov chain $X$.

From the uniform consistency over compact sets of the estimator $\hat{r}_{i,n}$, we have that if 
$nh_n/\log n \to \infty$ then, for all $y\in C$, $
\hat{g}_{i,n}(y) \to g_i(y)$, $\hat{f}_{i,n}(y) \to f_i(y)$ a.s. as $n\to \infty$.
This implies $\mathbb{E}[\hat{g}_{i,n}(y)|Y_{0:n}, \psi^{*}] \to g_i(y)$, and $\mathbb{E}[\hat{f}_{i,n}(y)|Y_{0:n}, \psi^{*}] \to f_i(y)$ as $n\to \infty$, thus we have $ \theta_i^{*}(y,Y_{0:n})\to r_i(y)$, {\it a.s.} as
$n\to \infty$.  As a consequence we obtain
\begin{equation*}
\lim_{n\to\infty}\lim_{t\to\infty}\theta_i^{t}(y,Y_{0:n})=r_i(y),\ a.s.
\end{equation*}
\end{rem}

In the following we will no explicit the dependence on the variables $y$ and $Y_{0:n}$; for instance, we will use  $\theta$ and $u(\theta)$ in order to denote $\theta(y)$ and $u(y,Y_{0:n},\theta)$.

We begin this task, assuming that $\theta^*=\lim_{t\to\infty}{\theta}^t$ is a stable stationary point,
that is $\theta^*\in \Theta$ satisfies $\nabla_\theta u(\theta^*)=0$ and
$\nabla_\theta^2 u(\theta^*)>0$. As $u(\theta)$ is a quadratic function with respect to $\theta$, then
\begin{equation*}
\nabla_\theta u(\theta^{t-1})= \nabla_\theta u(\theta^{*})+ \nabla_\theta^2 u(\theta^{*})(\theta^{t-1} - \theta^{*}).
\end{equation*}

Let $\varDelta_t=\theta^{t} - \theta^*$, since $\nabla_\theta u(\theta^*)=0$, from
equation \eqref{GradienteEst} we can write for $t\geq 1$
\begin{equation*}
\varDelta_t=\varDelta_{t-1}+\gamma_t\left(-\nabla_\theta u(\theta^{t-1})+\varsigma_t\right)
=\varDelta_{t-1}-\gamma_t\nabla_\theta^2 u(\theta^*)\varDelta_{t-1}+\gamma_t\varsigma_t
=\left(I-\gamma_t\nabla_\theta^2 u(\theta^*)\right)\varDelta_{t-1}+\gamma_t\varsigma_t.
\end{equation*}

Here, $\varDelta_t$ obeys a linear difference equation driven by the $(\mathcal{F}_{t})$-martingale difference sequence $\{\varsigma_t\}$, with
$\mathcal{F}_t$ the $\sigma$-algebra generated by $\{X_{1:n}^s\}_{s=0:t}$.
From the proof of the consistence result given in \cite[Theorem 3.1]{Rico-Lisandro-Luis}, we can verify that $\varsigma_t$ satisfies 
\begin{equation*}
\mathbb{E}(\|\varsigma_{t}\|^2|\mathcal{F}_{t-1})\leq \| \Psi(\theta^{t-1})\|^2,
\end{equation*}
where $\Psi(\theta)=(\Psi_1(\theta),\ldots,\Psi_m(\theta))$ and
\begin{equation*}
\Psi_i(\theta)= \frac{1}{nh}\sum_{k=0}^{n-1}(Y_{k+1}-\theta_i)K_h(y-Y_k).
\end{equation*}
By compactness $\|\Psi(\theta)\|^2$ is finite in any compact subset of $\Theta$. Thus, if the closure of $\{\theta^t\}_{t\geq0}$ is a compact subset of $\Theta$, we obtain
\begin{eqnarray}\label{M2Marting}
\sup_t \mathbb{E}(\|\varsigma_{t}\|^2|\mathcal{F}_{t-1}) < \infty.
\end{eqnarray}

%%%%%%%%%%%%%%% Remark 3.1
\begin{rem}\label{covBernoulli}
Since $\{X_{1:n}^{t}\}_{t\geq 1}$ 
is an ergodic Mar\-kov chain  with 
invariant distribution probability $p(X_{1:n}=x_{1:n}|Y_{0:n},\psi^*)$ then for any bounded measurable function $f:\{1,\ldots,m\}^n\to \mathbb{R}$, $\mathbb{E}(f(X_{1:n}^{t}))\to \mathbb{E}(f(X_{1:n})), \quad \mbox{ when } \quad t\to\infty$.
In particular,
\begin{eqnarray*}
\mathrm{Cov}(\car_i(X_{k+1}^{t}),\car_j(X_{k'+1}^{t})|\mathcal{F}_{t-1})&\to&
\chi_{i,j}(k,k'):=\mathrm{Cov}(\car_i(X_{k+1}),\car_j(X_{k'+1})),
\end{eqnarray*}
in probability when $t\to\infty$, where 
\begin{equation}\label{chi}
\chi_{i,j}(k,k')=
\left\{
\begin{array}{lll}
\mu_{i,k}^*(A^{*^{(k'-k)}}_{i,j}-\mu_{j,k'}^*) & \mbox{ if } & k\leq k',\\
\mu_{j,k'}^*(A^{*^{(k-k')}}_{j,i}-\mu_{i,k}^*) & \mbox{ if } & k'< k,
\end{array}
\right.
\end{equation}
with $\mu_{i,k}^*=p(X_{k+1}=i|Y_{0:n},\psi^*)$ the marginal probability measure of the invariant measure $p(X_{1:n}=x_{1:n}|Y_{0:n},\psi^*)$, and $A^{*^{(k'-k)}}_{i,j}$ is the $(i,j)$-th entry of the $(k'-k)$-th power of the transition matrix $A^*$ defined in Remark \ref{invmeasure}.
\end{rem}

The following lemma show that the martingale difference sequence $\{\varsigma_t\}$ satisfy the conditions for the Martingale Central Limit Theorem.

%%%%%%%%%%%%%%% Lemma 3.3 
\begin{lema}\label{cond_martingale}
Assume that the closure of $\{\theta^t\}$ is a compact subset of $\Theta$ and $\theta^*=\lim_{t\to\infty}\theta^t$ a.s. then,

\begin{enumerate}     
\item[i)] $\lim_{c\to\infty}\limsup_{t\to\infty}\mathbb{E}\left(\|\varsigma_t\|^2
\car\left(\|\varsigma_t\|>c\right)|\mathcal{F}_{t-1}\right)= \; 0$,
in probability.
\item[ii)] $\lim_{t\to\infty}\mathbb{E}\left(\varsigma_t\varsigma_t^{T}|\mathcal{F}_{t-1}
\right)= \Gamma$ in probability, where $\Gamma$ is the covariance matrix defined by
\begin{equation*}  
\Gamma_{i,j} =\frac{4}{n^2h^2}\sum_{k,k'=0}^{n-1}(Y_{k+1}-\theta_i^{*})(Y_{k'+1}-\theta_j^{*})K_h(y-Y_k)K_h(y-Y_{k'})\chi_{i,j}(k,k'),
\end{equation*}
where $\chi_{i,j}(k,k')$ is given by \eqref{chi}.
\end{enumerate}
\end{lema}

The proof of this lemma is reported to Appendix \ref{appA}. 

We now establish the CLT for the estimated mean error of algorithm,
\begin{equation*}
\bar{\theta}^{t}-\theta^*= \frac{1}{t}\sum_{k=1}^{t}\varDelta_k.
\end{equation*}

As pointed out by \cite{cappe-moulines-ryden}, using the asymptotic 
law given in Theorem \ref{tasadeconvergencia} has the drawback to require in practice the
estimate of $\Gamma$. Good choices of $\gamma_t$ - of the form $\gamma_t=\gamma t^{-\alpha}$ for $1/2<\alpha<1$ - 
allow to  achieve the rate $ t^{-1/2}$ for $\varDelta_t$. 
In this case, $\alpha$ should be as small as possible in order to increase the rate
of the numerical algorithm used to calculate $\bar{\theta}_t$.

%%%%%%%%%%%%%%%%%%%%%%%%%%%%%%%%%%%%%%%%%%%%%%%%%%%%%% Section 4 

\section{Numerical studies}\label{sect.4}
In this section we illustrate the performance of the algorithm developed, in the previous section, by applying it to a simulated data. The first subsection deals with the asymptotic normality obtained in the fully observed data case (Theorem \ref{normalidadTO}), whereas the second looks at the asymptotic normality obtained in the partially observed data case (Theorem \ref{tasadeconvergencia}).

We consider an MS-NAR model with $m=3$. The autoregressive functions are 
\begin{eqnarray*}
r_1(y)&=&0.7y+2e^{(-25y^2)},\\
r_2(y)&=&\frac{2}{1+e^{10y}}-1,\\
r_3(y)&=& -2\cos(y)-1.
\end{eqnarray*}

We take $\Phi$, the density probability function of the white noise $\{e_k\}$, to be a Gaussian density with zero mean and variance $\sigma^2=0.25$. The transition probability  matrix of Mar\-kov chain $\{X_k\}$ is given by
\begin{equation*}
A=\left(\begin{array}{ccc}
 0.9500 & 0.0250 & 0.0250\\
0.0250  & 0.9500 & 0.0250\\
0.0250  & 0.0250 & 0.9500
\end{array}\right).
\end{equation*}

We simulate a sample path of $\{Y_0,(Y_k,X_k)\}_{k=1:n}$ of length $n=3000$. The time series data $Y_{0:n}$ is shown in Figure \ref{fig1}.
\begin{figure}[h]
\centering
\includegraphics[width=0.7\textwidth]
{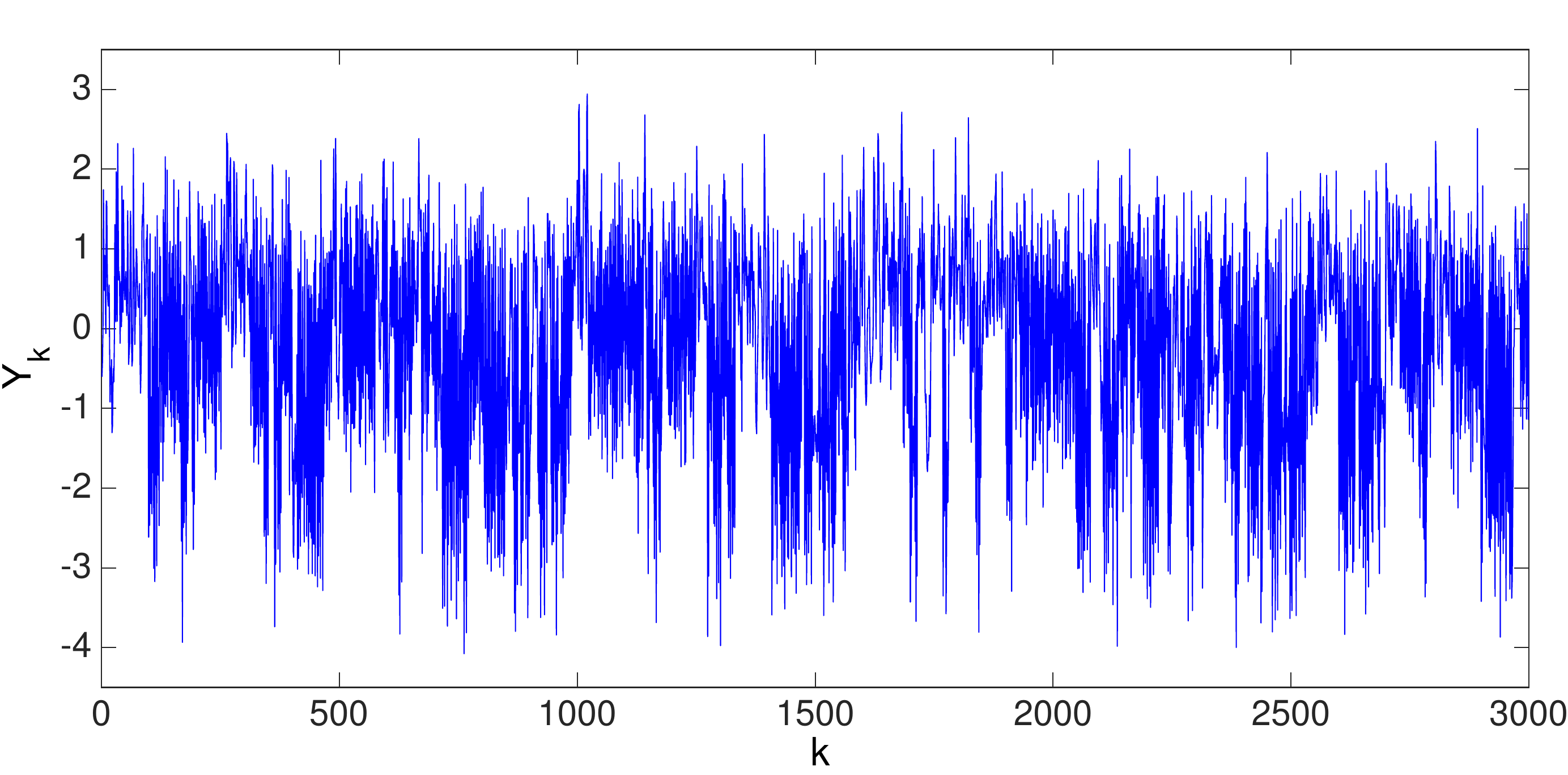}
\caption{\small{Simulated sample path $Y_{0:n}$ from an MS-NAR model with $m = 3$ states and $n = 3000$.}
}\label{fig1}
\end{figure}

%%%%%%%%%%%%%%%%%%%%%%%%%%%%%%%%%%%%%%%%%%%%%%%%%%%%%% Section 4.1 

\subsection{Fully observed data case}

Here, we assume that the complete data $\{Y_{0:n},X_{1:n}\}$ is available. Figure \ref{fig2} shows the scatter plot of $Y_k$ versus $Y_{k-1}$ with real state given by $X_k$. 
\begin{figure}[h]
\centering
\includegraphics[width=0.7\textwidth]{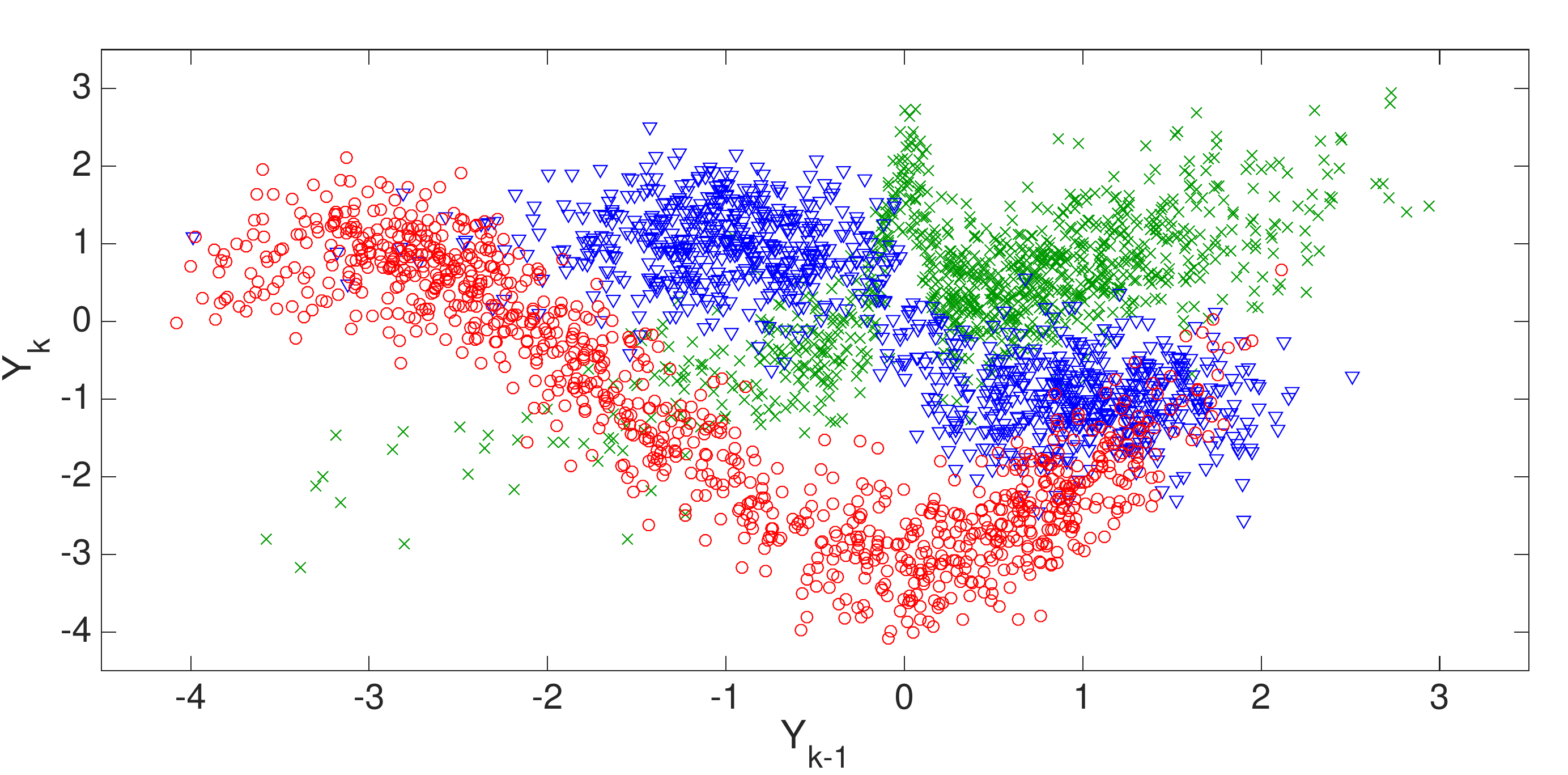}
\caption{\small{Scatterplot $Y_{k}$ versus $Y_{k-1}$ with real state classification. The points are labeled with respect to the real state of $X_k$: cross for $X_k=1$,  triangles for $X_k=2$, and circles for $X_k=3$.}
}\label{fig2}
\end{figure}

We consider the Nadaraya-Watson kernel estimators $\hat{r}_{i,n}$ defined in \eqref{NWEstim}, we use the ``Triweight'' kernel function. We take the bandwidth parameter $h_{n,i}=k_0\hat{\sigma_i}(\log(n)/n)^{1/5}$ for the state $i$, where $k_0=0.55$ and $\hat{\sigma}_i^2$ is the variance  estimate of $Y_k$ given $X_k=i$. In this simulation, we have $\hat{\sigma}_1^2=0.2630$, $\hat{\sigma}_2^2=0.2441$, and $\hat{\sigma}_3^2=0.2386$.

In Figure \ref{fig3}, we show how good the estimation is when dealing with samples of finite size which is the case in practice. We show the Nadaraya-Watson kernel estimators $\hat{r}_{i,n}$ along with the real functions $r_i$. 
\begin{figure}[h]
\centering
\includegraphics[width=0.7\textwidth]{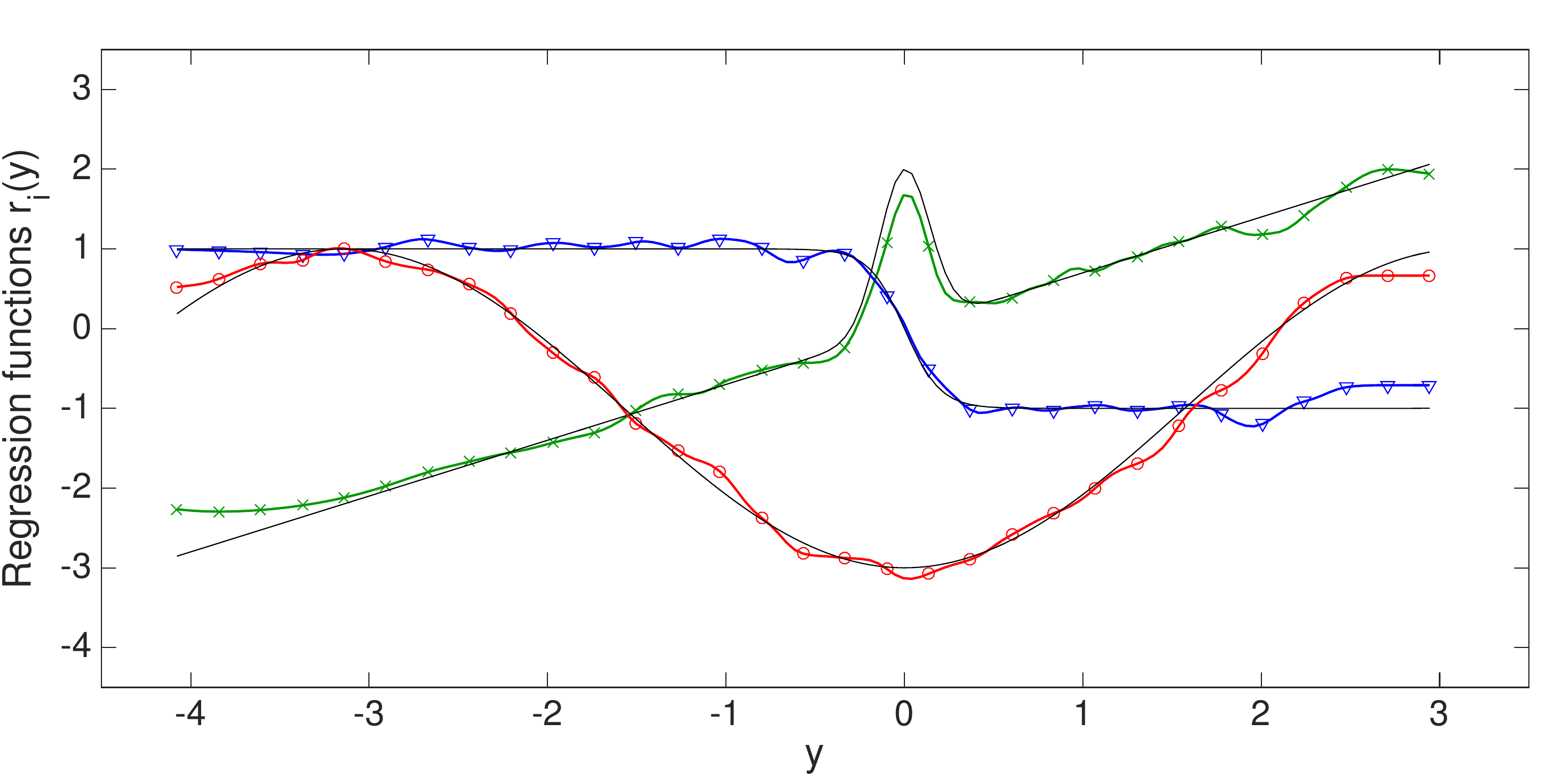}
\caption{\small{Non-parametric kernel regression function estimates in the fully observed data case. The real functions $r_i$ are shown with solid lines and the Nadaraya-Watson kernel estimators with: line with cross for $\hat{r}_{1,n}$, line with triangles for $\hat{r}_{2,n}$, line with circles for $\hat{r}_{3,n}$.}
}\label{fig3}
\end{figure}

The residual variance function, given in Theorem \ref{normalidadTO}, is estimated in the following way: the term $G_{2i}(y)/f_i(y)$ by $\hat{\sigma}^2_i$, the density function $f_i(y)=\mu_i p(Y_0=y)$ is estimate considering the non-parametric kernel density estimation $\hat{p}_0(Y_0=y)$ of the data $Y_{0:n}$ shown in Figure \ref{fig5}, and $\hat{\mu}_i=\frac{1}{n}\sum_{k=1}^n\car_{i}(X_k)$. For this simulation we have obtained $\hat{\mu}_1=0.3321$, $\hat{\mu}_2=0.3511$, and  $\hat{\mu}_3=0.3168$.  
\begin{figure}[h]
\centering
\includegraphics[width=0.7\textwidth]{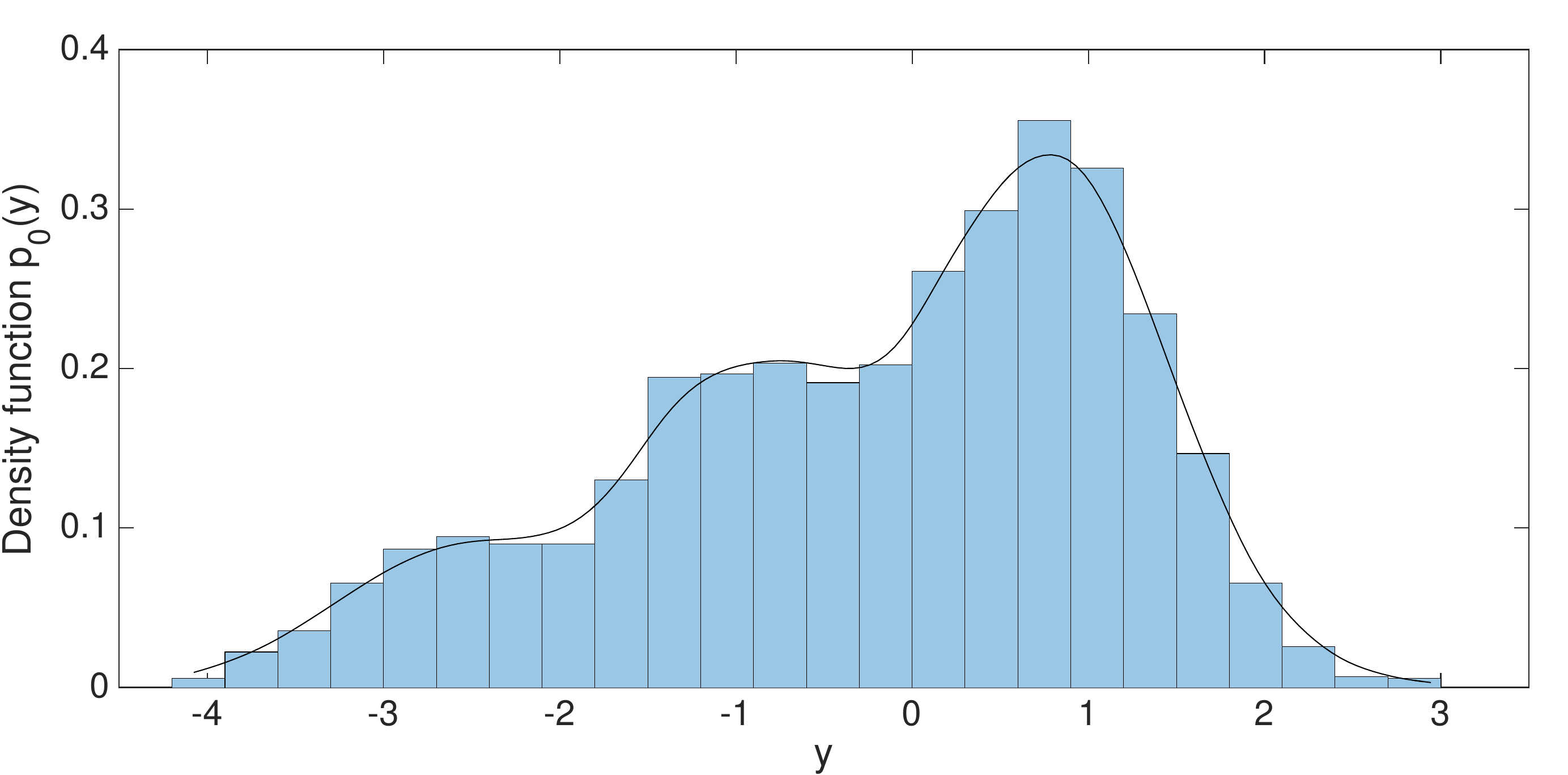}
\caption{\small{Non-parametric kernel density function estimate $\hat{p}_0(y)$ and the frequency histogram of $Y_{0:n}$.}
}\label{fig5}
\end{figure}

For the Triweight kernel function we have $\|K\|^2_2=350/429$. 
Thus, the $(1-\alpha)$-pointwise confidence interval at $y$ is estimated by 
$$\hat{r}_{i,n}(y) \pm \mathit{z}_{(1-\frac{\alpha}{2})}\frac{\hat{\sigma}_i\|K\|_2}{\sqrt{n h_n \hat{f}_{i,n}(y)}},$$ where $\alpha=0.05$ and $\mathit{z}_{(1-\frac{\alpha}{2})}$ is the $(1-\frac{\alpha}{2})$-quantile of the standard normal distribution.
In Figure \ref{fig4}, we show the Nadaraya-Watson kernel estimators $\hat{r}_{i,n}$ and the $95\%$ pointwise confidence bands determined from the asymptotic normality given in Theorem \ref{tasadeconvergencia}.
\begin{figure}[h]
\centering
\subfigure{\includegraphics[width=0.7\textwidth, height=4cm]{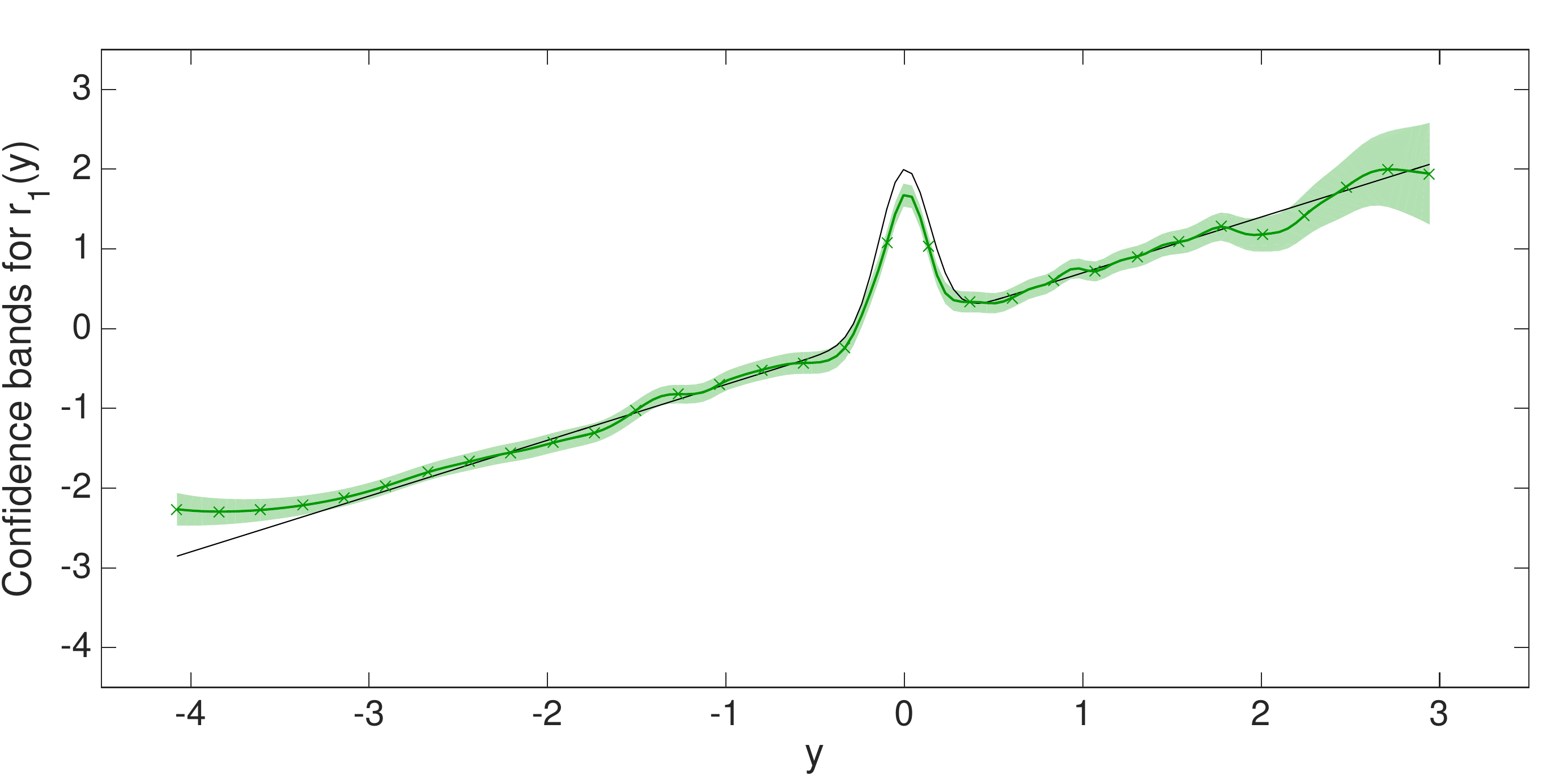}}
\subfigure{\includegraphics[width=0.7\textwidth, height=4cm]{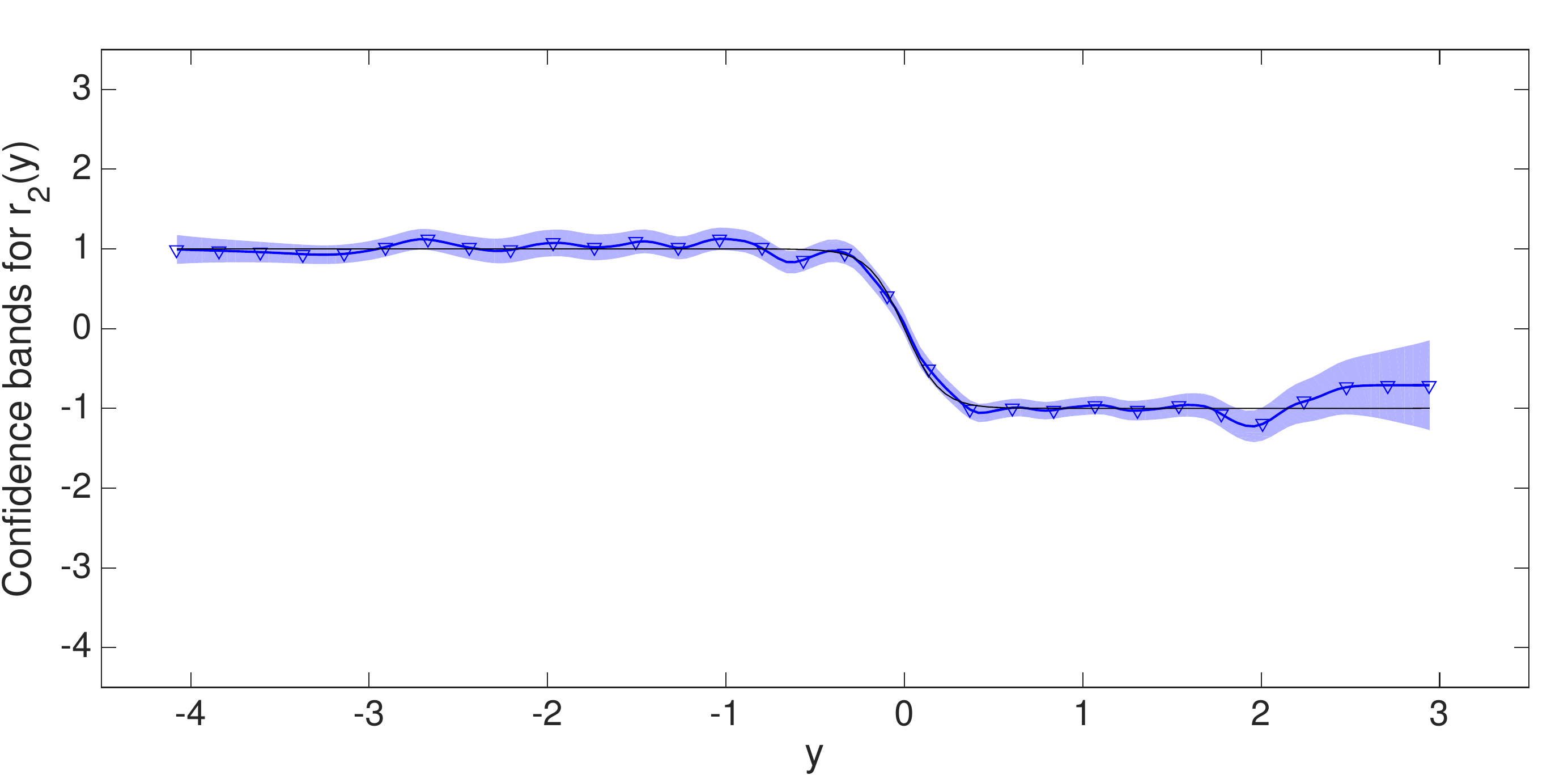}}
\subfigure{\includegraphics[width=0.7\textwidth, height=4cm]{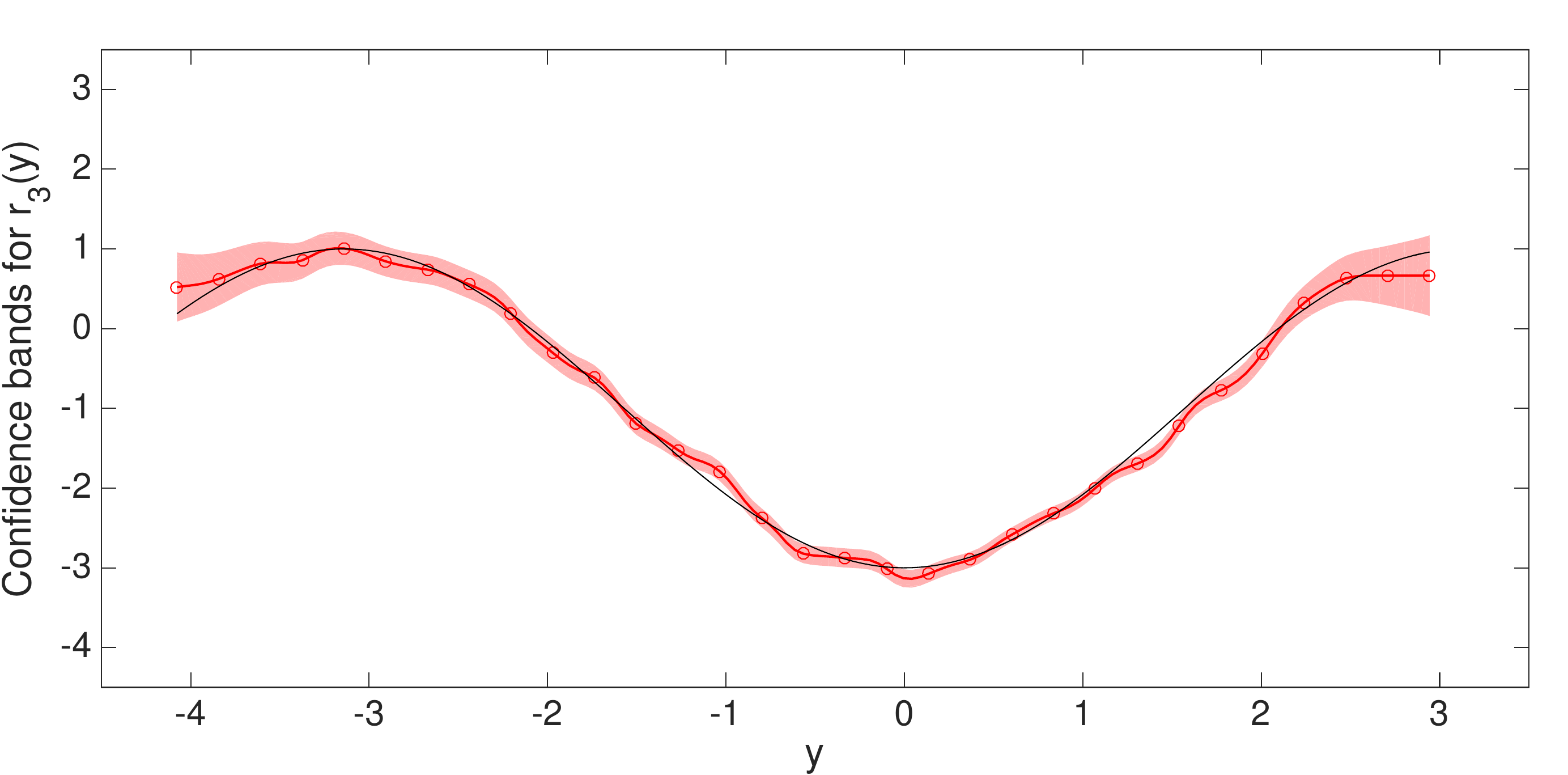}}
\caption{\small{Non-parametric kernel regression function estimates in the fully observed data case. The Nadaraya-Watson kernel estimators are shown with: line with cross for $\hat{r}_{1,n}$, line with triangles for $\hat{r}_{2,n}$, line with circles for $\hat{r}_{3,n}$. The real functions $r_i$ are shown with solid lines, and the fill areas correspond to the $95\%$ pointwise confidence bands.}
}\label{fig4}
\end{figure}

Finally, we show in Figure \ref{fig6} the non-parametric kernel estimation of the conditional density probability functions $p(Y_k = y | X_k = i)$.
\begin{figure}[h]
\centering
\includegraphics[width=0.7\textwidth]{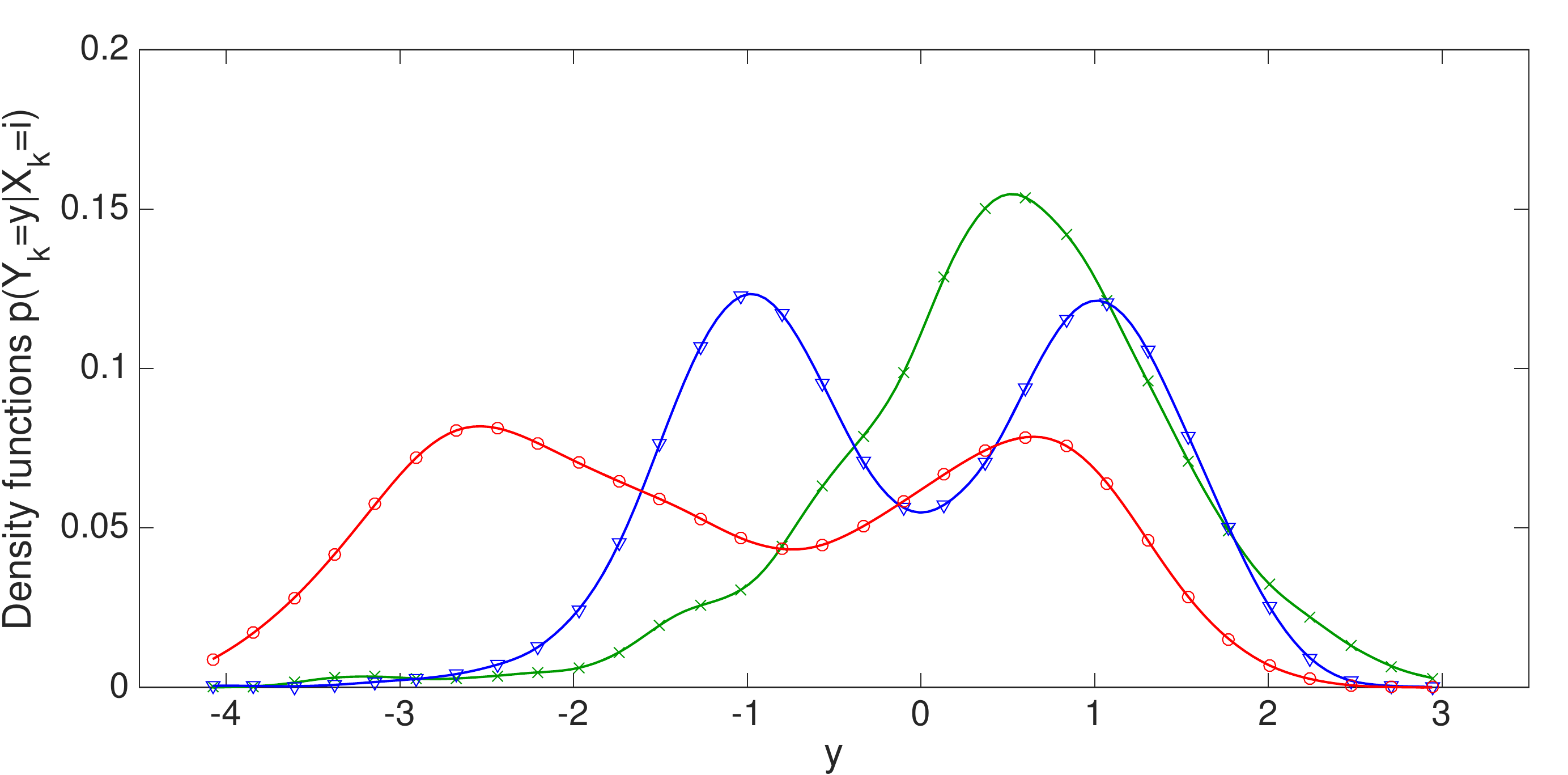}
\caption{\small{Non-parametric kernel estimation of  $p(Y_k = y | X_k = i )$}.
}\label{fig6}
\end{figure}

%%%%%%%%%%%%%%%%%%%%%%%%%%%%%%%%%%%%%%%%%%%%%%%%%%%%%% Section 4.2 

\subsection{Partially observed data case}

In this section, we implemented our Robbins-Monro algorithm for the data $Y_{0:n}$ considering that $X$ is hidden. 

We  take the ``Triweight'' kernel function $K$, the bandwidth parameter  $h_{n,i}=k_0\hat{\sigma_i^t}(\log(n)/n)^{1/5}$, where $k_0=0.55$ and $\hat{\sigma}_i^{(t)}$ is the standard deviation estimate of $Y_k$ given $X^{t}_k=i$ for the iteration $t$. The smoothing step $\gamma_t = t^{-0.6}$, and the initial estimates for the Markov chain $X_{1:n}^0$ in the Step 0 are assumed uniform random variables. Due to the complexity of the regression functions in this example, the estimation of an MS-AR model is not a good initial estimate. 
For $T=1000$ iterations of our Robbins-Monro algorithm, we show the following results.

The estimated matrix transition is
\begin{equation*}
\hat{A}=\left(\begin{array}{ccc}
 0.9145  & 0.0494 & 0.0361\\
0.0528  & 0.8729  & 0.0743\\
0.0305  &  0.0912 & 0.8783
\end{array}\right).
\end{equation*}

Figure \ref{fig7} shows the square error for the estimate $A^t$. We can note that the convergence is quickly reached for the probability transition matrix $A^t$.
\begin{figure}[h]
\centering
\includegraphics[width=0.7\textwidth]{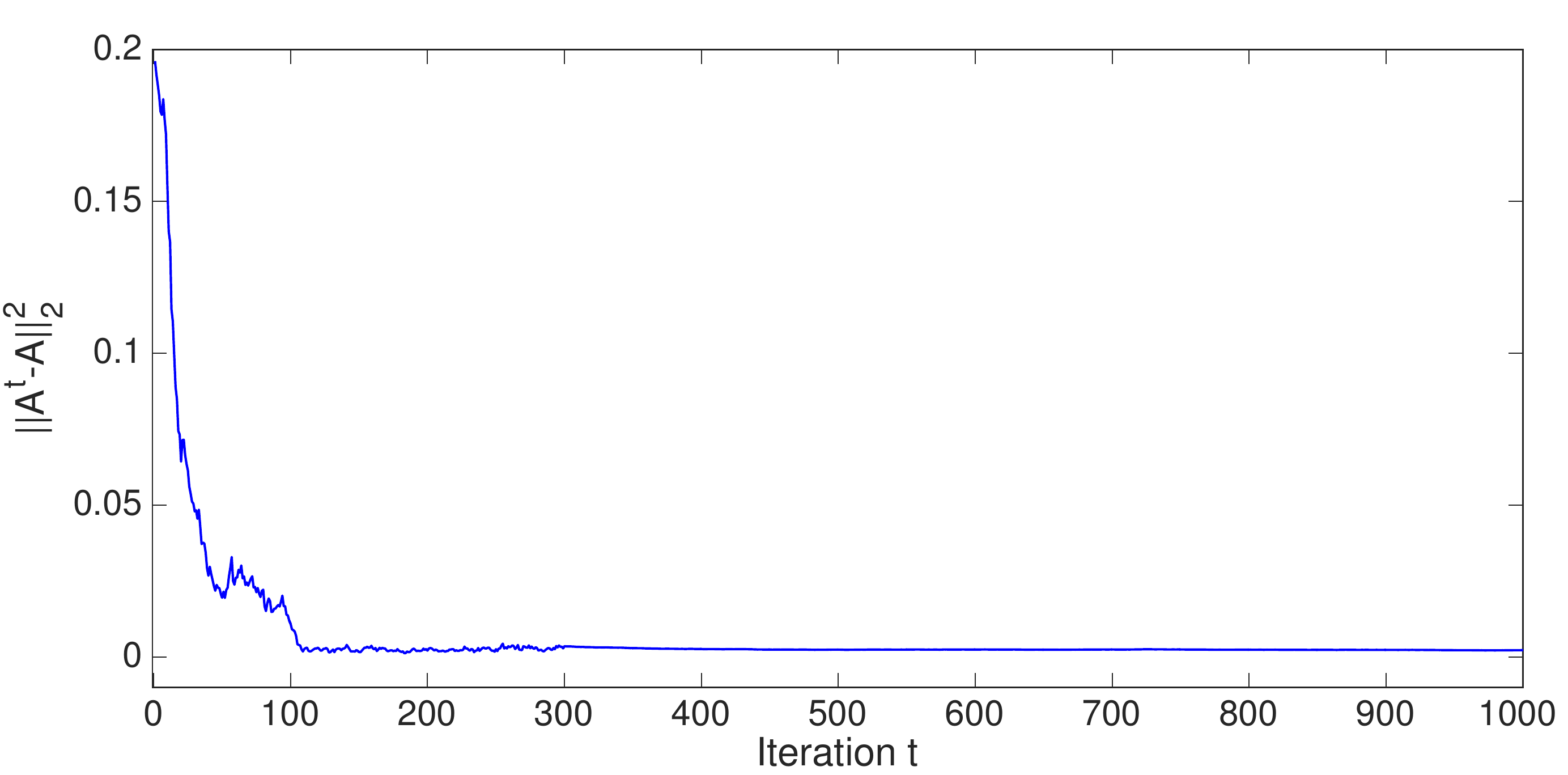}
\caption{\small{The square estimation error $\|A^t-A\|_2^2$, for iteration $t=1:T$.}
}\label{fig7}
\end{figure}

We recall that on the Step R of our algorithm, we restore the unobserved data by drawing a sample $X_{1:n}^{t}$ from the conditional distribution $p(X_{1:n}|Y_{0:n},\psi^{t-1})$.
 Figure \ref{fig8} shows the marginal probability measures $p(X_{k}=i|Y_{0:n},\psi^{t-1})$ of the the conditional distribution $p(X_{1:n}|Y_{0:n},\psi^{t-1})$, for iteration $t=1000$.
\begin{figure}[h]
\centering
\subfigure{\includegraphics[width=0.7\textwidth, height=4cm]{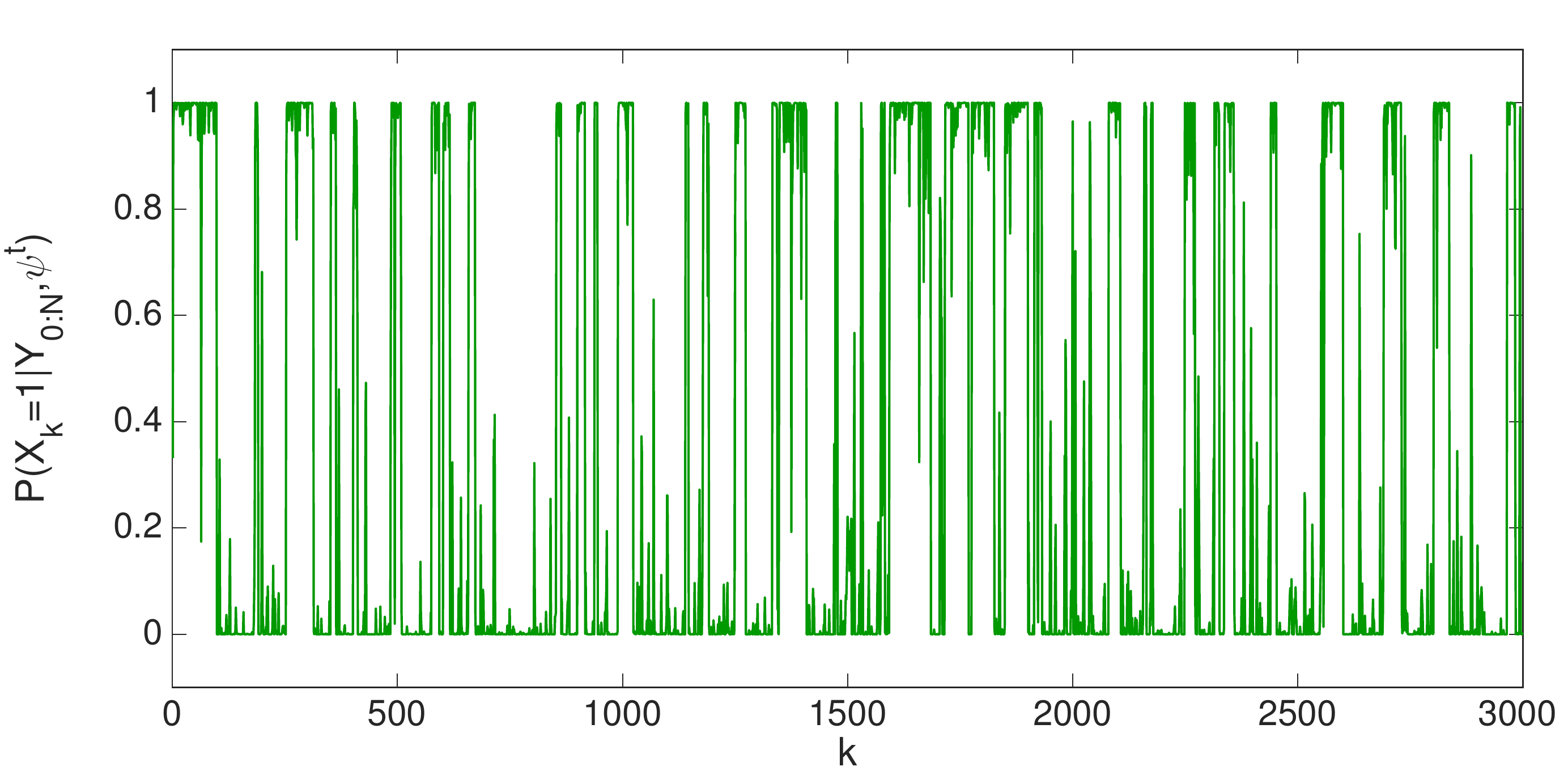}}
\subfigure{\includegraphics[width=0.7\textwidth, height=4cm]{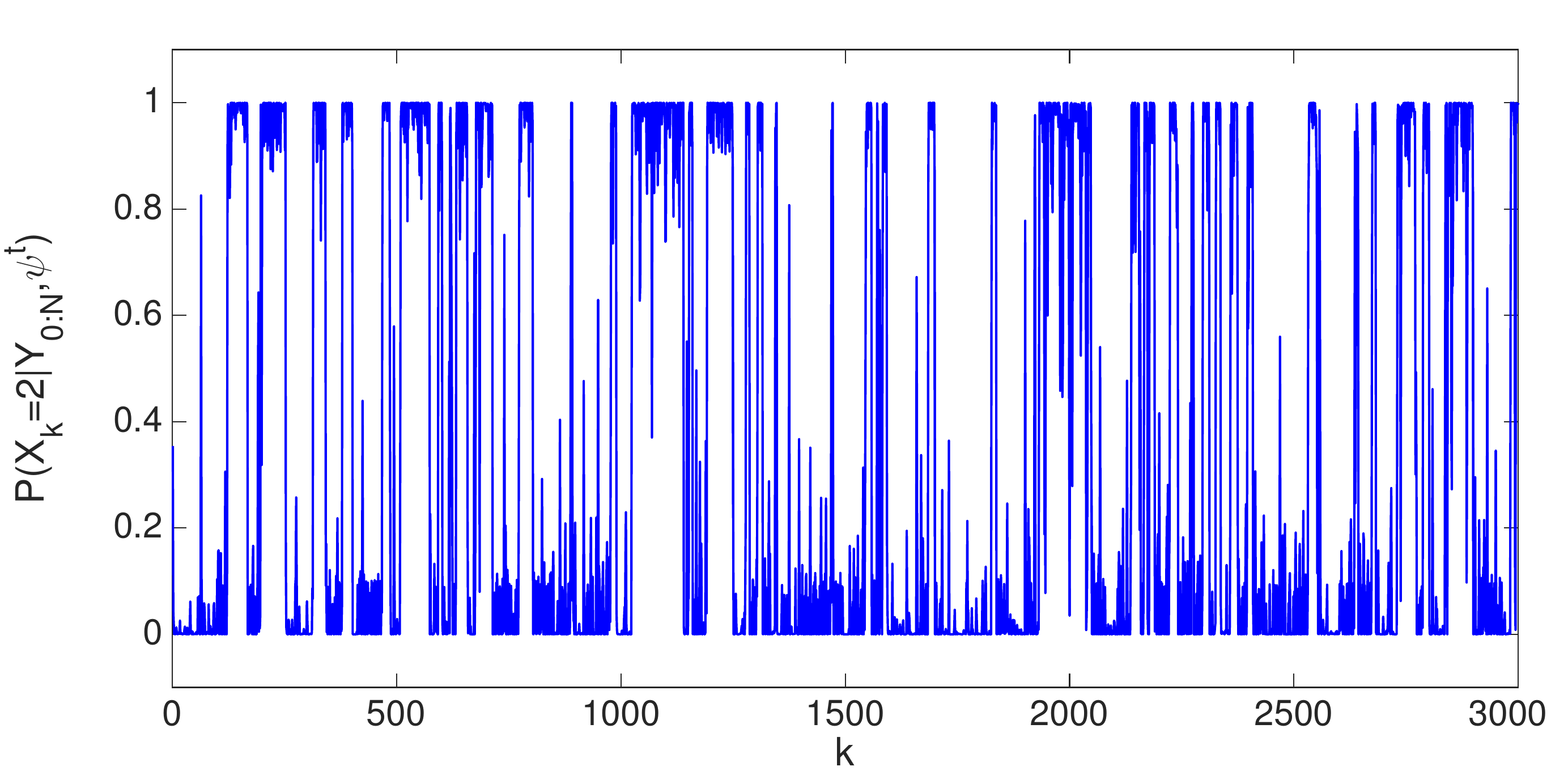}}
\subfigure{\includegraphics[width=0.7\textwidth, height=4cm]{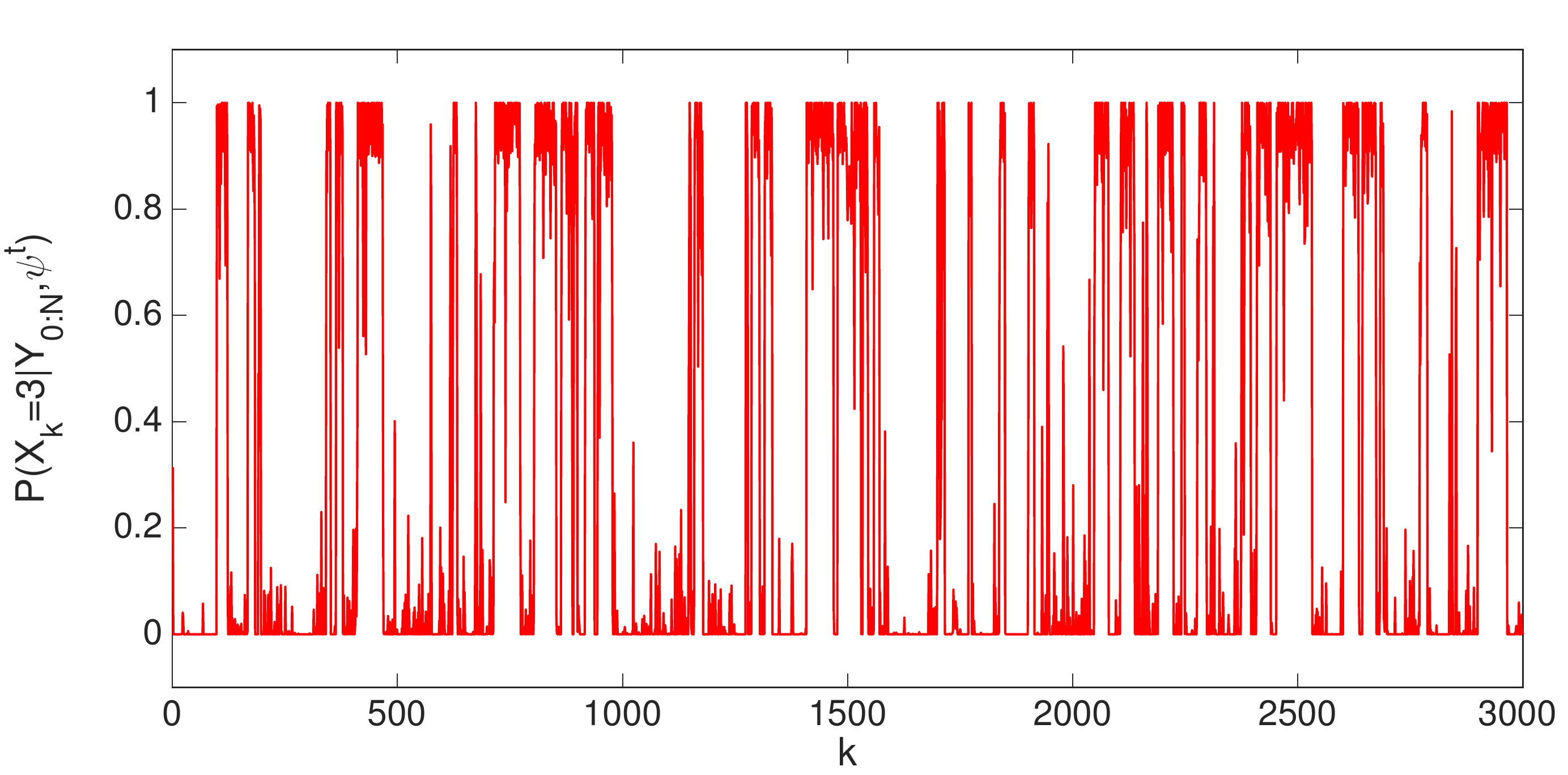}}
\caption{\small{Conditional probability  of Markov regime $X_k$ given the observed data $Y_{0:n}$.}
}\label{fig8}
\end{figure}

The probability of classify $X_k^t$ on state $i$ given that the real data $X_k$ is on state $j$ is estimated by: 
\begin{eqnarray*}
P(X_k^t=i | X_k =j ) &=& \frac{1}{n_j}\sum_{k=1}^n \car(X_k^t=i, X_k=j)\\
 &=&\left(\begin{array}{ccc}
    0.9435   \;& 0.0348  \;&  0.0413 \\
    0.0370   \;& 0.9177  \;&  0.0705 \\
    0.0195   \;& 0.0474  \;&  0.8882
\end{array}\right),
\end{eqnarray*}
where $n_j=\sum_{k=1}^n \car(X_k=j)$. 

Thus, the misclassification rate is 
\begin{equation*}
P(X_k^t\neq X_k)=  1- \frac{1}{n}\sum_{k=1}^n \sum_{i=1}^m \car(X_k^t=X_k=i)= 0.0837,
\end{equation*}
i.e $8.37\%$ of the data $X$ have been misclassified.

Figure \ref{fig9} shows the scatter plot of $Y_k$ versus $Y_{k-1}$ with estimated state $X^t_{1:n}$.

\begin{figure}[h]
\centering
\includegraphics[width=0.7\textwidth]{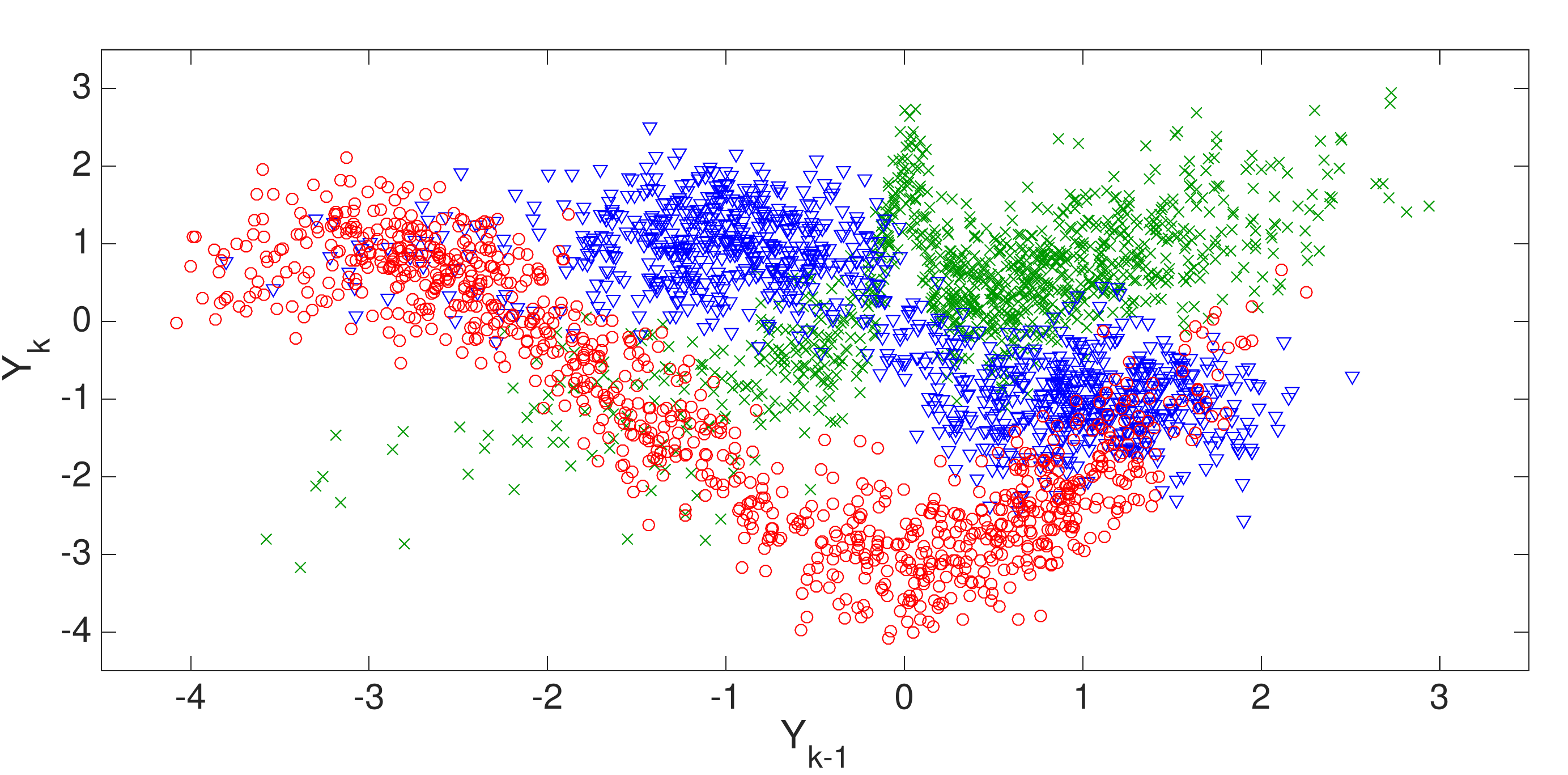}
\caption{\small{Scatter plot of $Y_{k}$ versus $Y_{k-1}$ with estimated state classification obtained by our Robbins-Monro algorithm. The points are labeled with respect to the estimated state of $X$: cross for $X_k^t=1$, triangles for $X_k^t=2$, and circles for $X_k^t=3$.}
}\label{fig9}
\end{figure}

\vspace{0.7cm}
The non-parametric regression functions obtained by the Robbins-Monro procedure are shown in Figure \ref{fig10}.

\begin{figure}[h]
\centering
\includegraphics[width=0.7\textwidth]{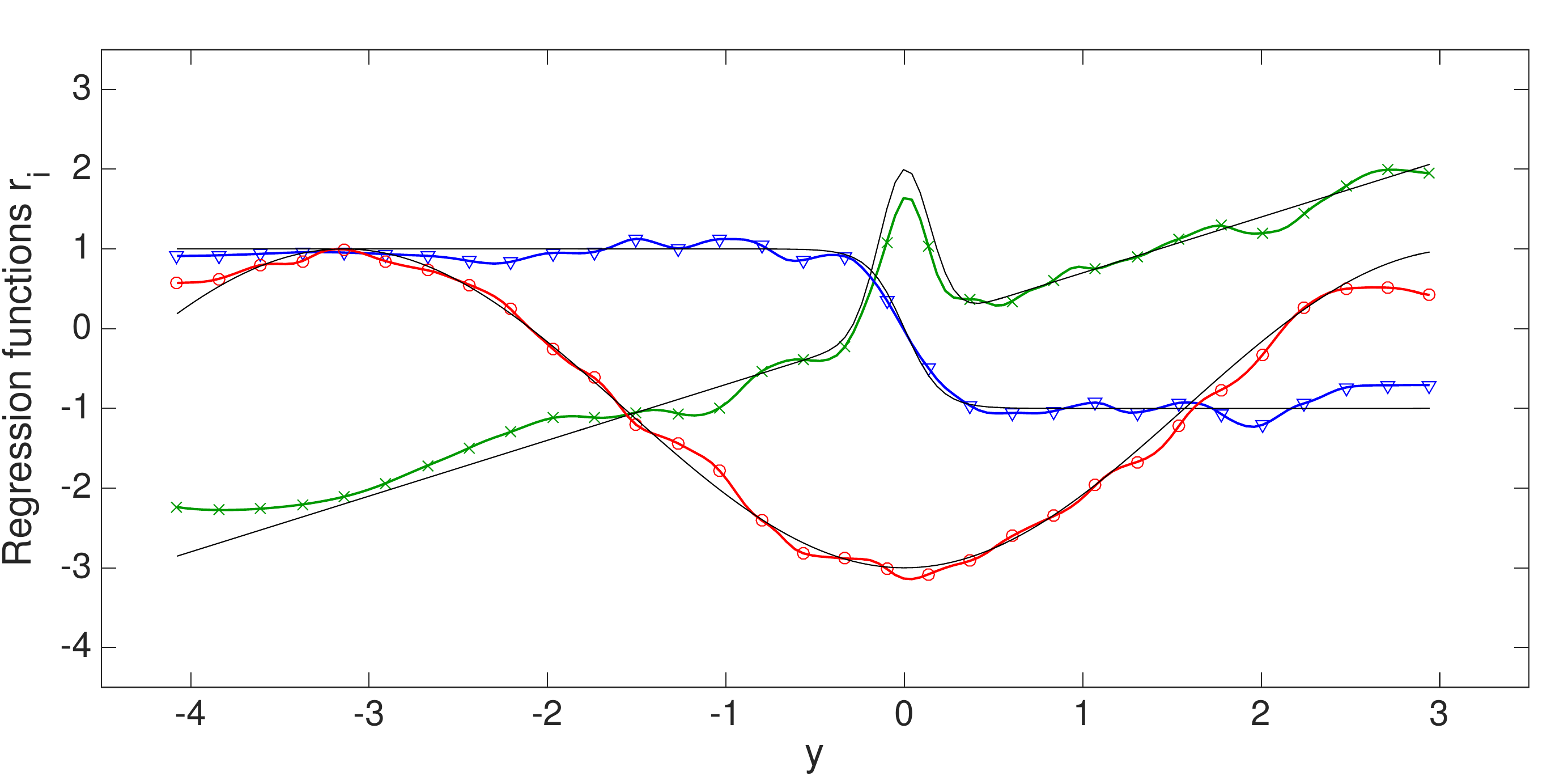}
\caption{\small{Non-parametric estimation of regression functions $r_i$, for i=1:m, obtained by our Robbins-Monro algorithm. The real functions $r_i$ are shown with solid lines and the estimates with: line with cross for $\hat{r}_{1,n}$, line with triangles for $\hat{r}_{2,n}$, line with circles for $\hat{r}_{3,n}$.}
}\label{fig10}
\end{figure}

In Figure \ref{fig11}, we show the non-parametric regression function estimates $\hat{r}^{t}_{i,n}$ obtained by the Robbins-Monro algorithm along with the Nadaraya-Watson kernel estimators $\hat{r}^{*}_{i,n}$. Moreover, we display the $95\%$ pointwise confidence bands determined from the asymptotic normality given in Theorem \ref{tasadeconvergencia}.
\begin{figure}[ht]
\centering
\subfigure{\includegraphics[width=0.7\textwidth, height=4cm]{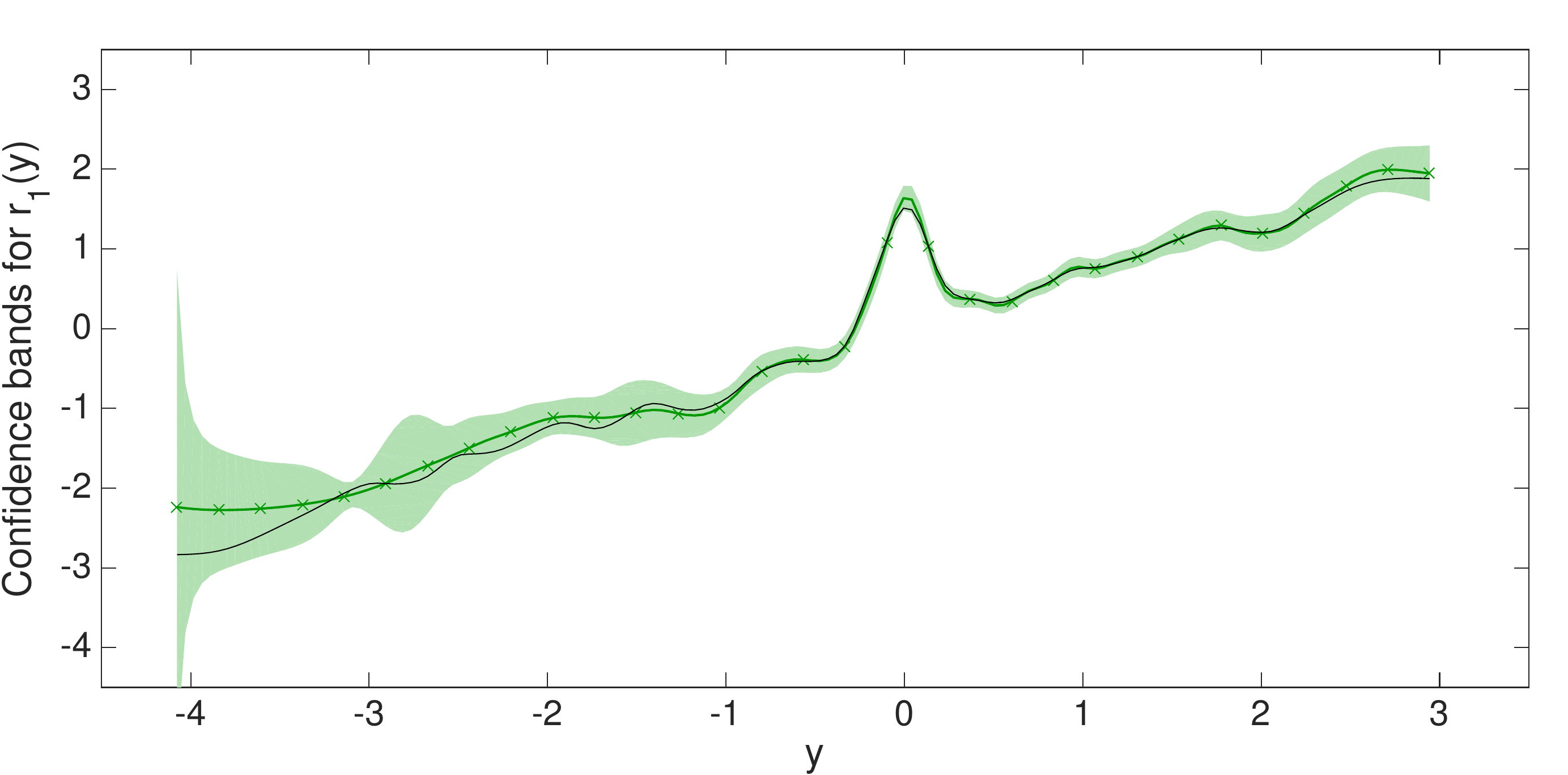}}
\subfigure{\includegraphics[width=0.7\textwidth, height=4cm]{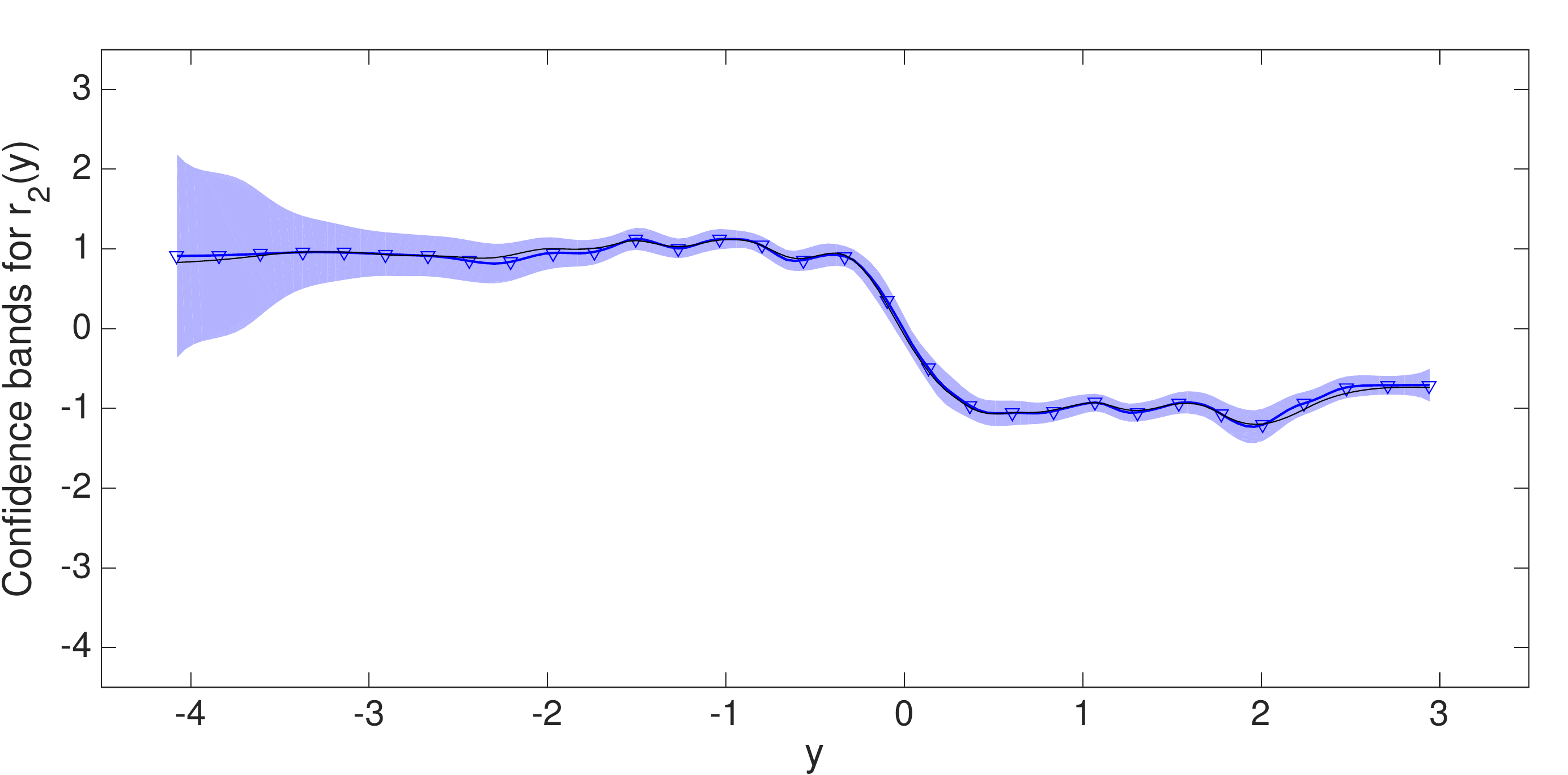}}
\subfigure{\includegraphics[width=0.7\textwidth, height=4cm]{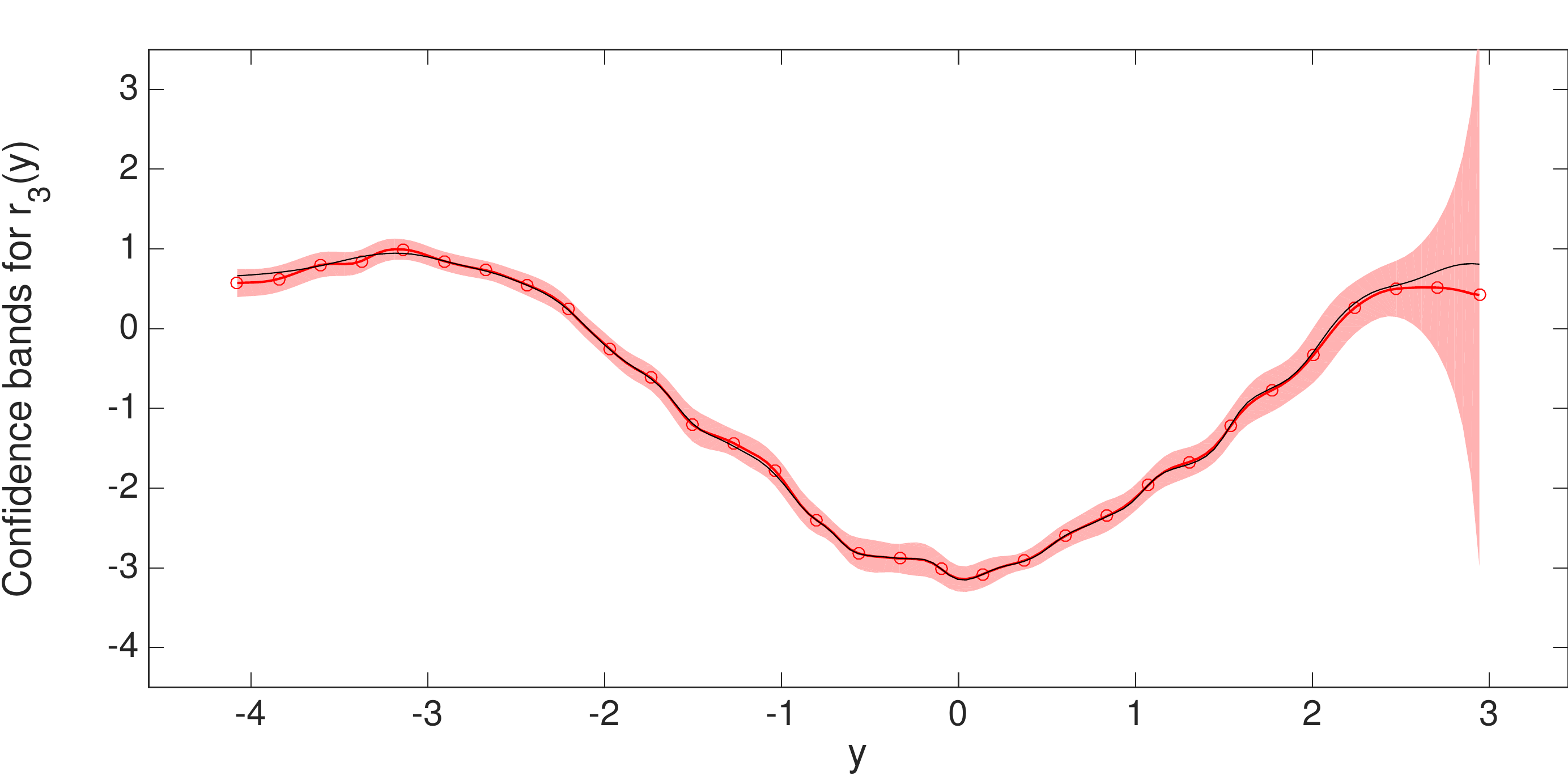}}
\caption{\small{Non-parametric regression function estimates  obtained by means of the Robbins-Monro algorithm. The non-parametric regression function estimates $\hat{r}^{t}_{i,n}$ are shown with: line with cross for $\hat{r}^t_{1,n}$, line with triangles for $\hat{r}^t_{2,n}$, line with circles for $\hat{r}^t_{3,n}$. The Nadaraya-Watson kernel estimators $\hat{r}^{*}_{i,n}$} are shown with solid lines, and the fill areas correspond to the $95\%$ pointwise confidence bands.
}\label{fig11}
\end{figure}

The residual variance function given by Theorem \ref{tasadeconvergencia} is estimated in the following way: First, replacing $\psi^*=(A^*, \theta^*)$ by $\psi^t=(A^t, \theta^t)$, and $p(X_k|Y_{0:n},\psi^*)$ by $p(X_k|Y_{0:n},\psi^t)$. Second, we compute the Nadaraya-Watson kernel estimators $\theta^*$ as in Remark \ref{ThetaAster}. Third, we evaluate the second derivatives of $u(y,Y_{0:n},\theta)$ at $\theta^*$ to obtain the matrix $\nabla_\theta^2 u(y,Y_{0:n},\theta^*)$. We compute the covariance matrix $\Gamma$ using  the results given in Lemma \ref{cond_martingale} and Remark \ref{covBernoulli}. Finally, the $(1-\alpha)$-pointwise confidence interval at $y$ is estimated from the asymptotic normality given in Theorem \ref{tasadeconvergencia}.

Figure \ref{fig12} shows the residual variance estimates $\hat{\sigma}_i^2(t)$. For the last iteration, t=1000, we have obtained $\hat{\sigma}_1^2(t)= 0.2707$, $\hat{\sigma}_2^2(t)=0.2500$, and $\hat{\sigma}_3^2=0.2483$. In this figure, we can see that the convergence is quickly approaching, considering that the real residual variance is $\sigma^2=0.25$.

\begin{figure}[h]
\centering
\includegraphics[width=0.7\textwidth]{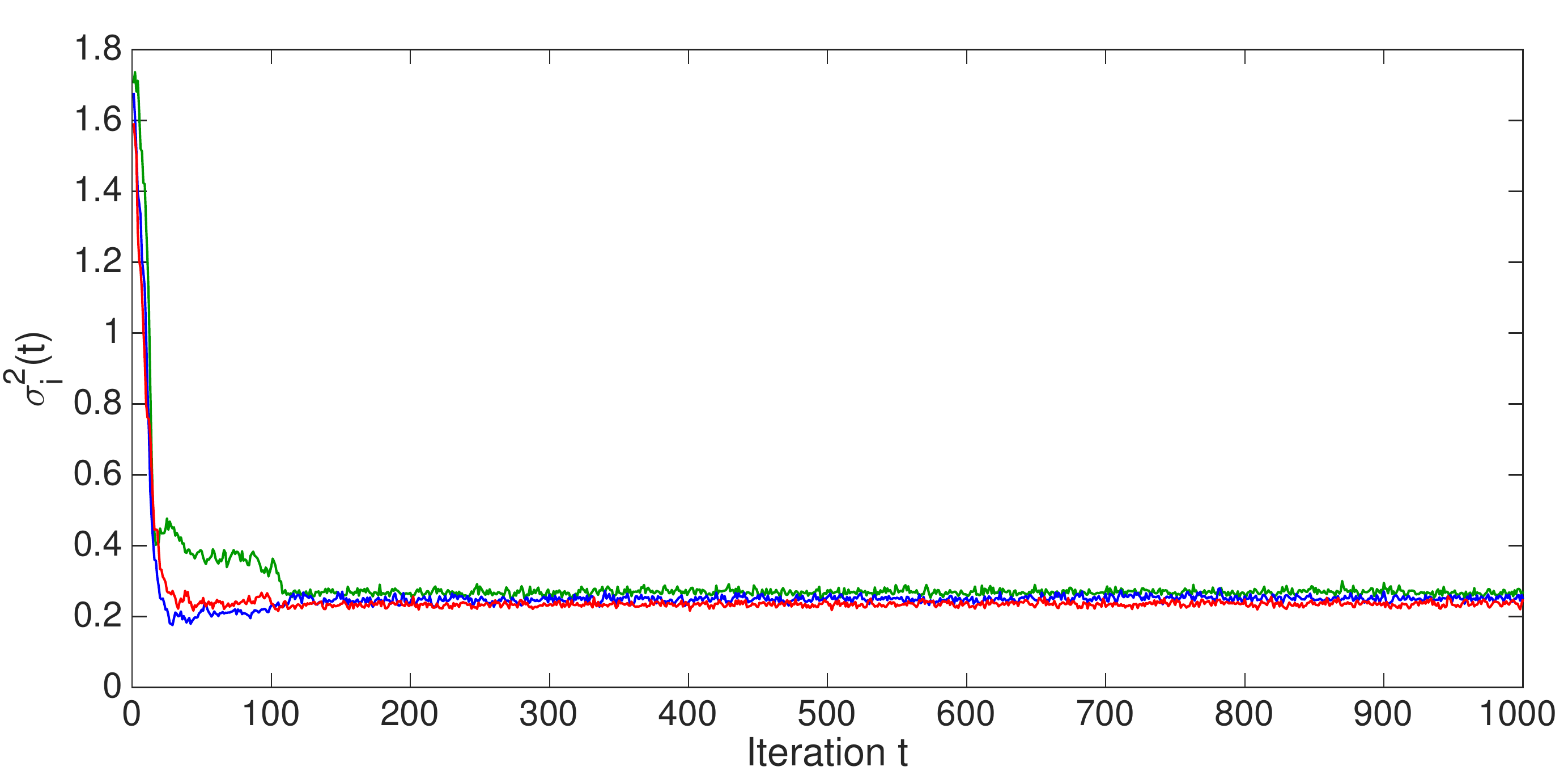}
\caption{\small{Residual variance estimates $\sigma_i^2(t)$, for iteration $t = 1:T$ and state $i = 1:m$}
}\label{fig12}
\end{figure}

Finally, we show in Figure \ref{fig13} the non-parametric kernel estimation of conditional density probability functions $\hat{f}_{i,n}(y)=p(Y_k = y | X_k^t = i)$. Note that $\nabla_\theta^2 u(y,Y_{0:n},\theta^*)$ is a diagonal matrix with $i$-th diagonal component given by $2f_i(y)$, thus for the values of $y$ where $\hat{f}_{i,n}(y)$ is nearly to zero we have larger pointwise confidence intervals.   
\begin{figure}[h]
\centering
\includegraphics[width=0.7\textwidth]{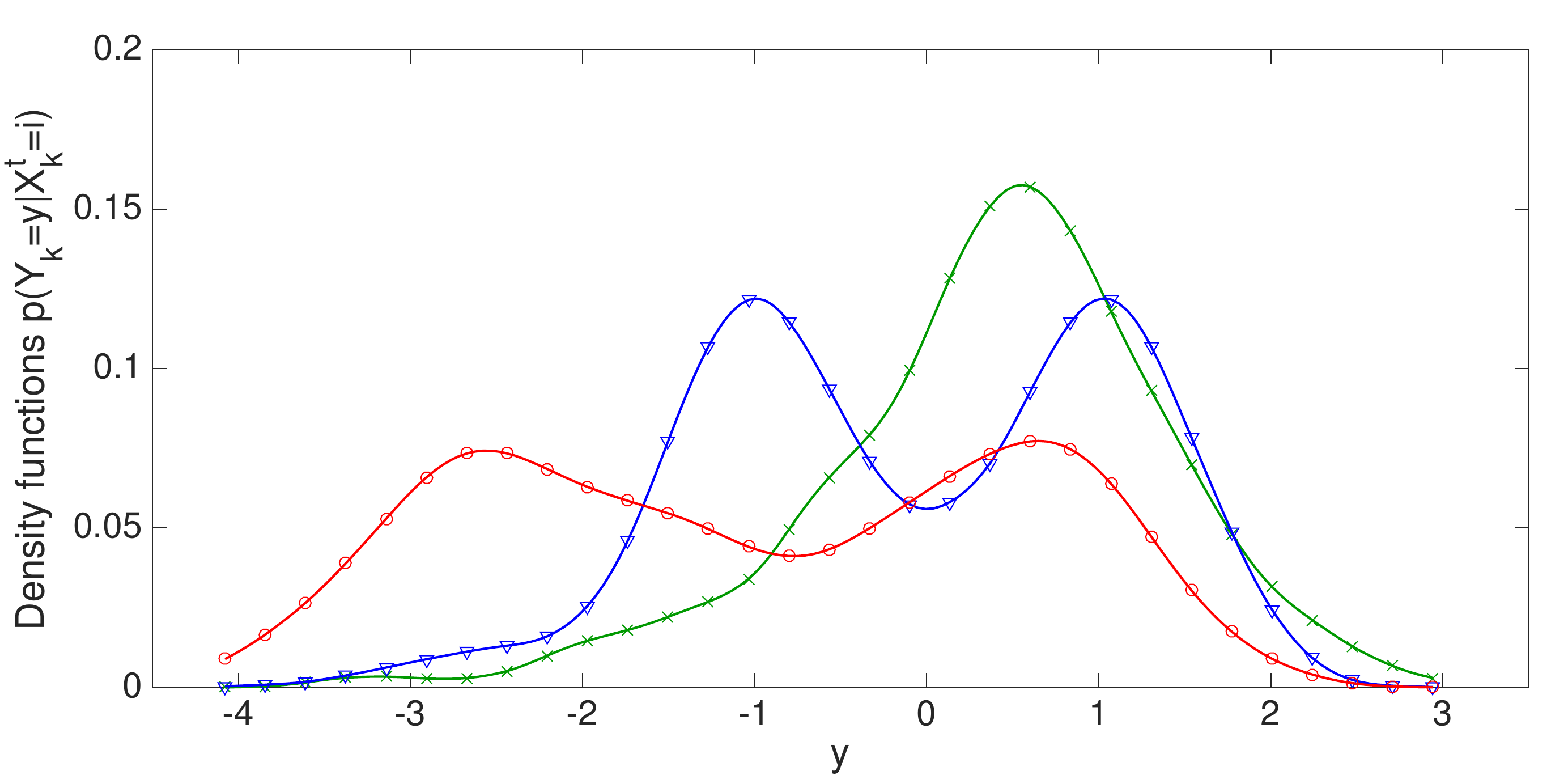}
\caption{\small{Non-parametric kernel estimation of  $p(Y_k = y | X_k^t = i )$.}
}\label{fig13}
\end{figure}

Notice that the algorithm has a good  performance. Moreover, convergence  is quickly reached for  the probability transition matrix $A^t$. It is seen  that the algorithm has difficulties in identifying the state of the Markov chain at intersection points of the regression functions. A loss of numerical identifiability occurs at these points due to the discretization step and the sample size. The misclassification rate is less than $10\%$ of the data, however the regression functions are well estimated.
A misclassification of data arises if the variance is large with respect to the range of regression functions. 

The variance of the regression function estimators increases at the edges of its supports, this is because there are few data at the edges and the kernel estimation is more sensitive to the choice of the bandwidth $h_n$. For this reason it is advisable to combine our algorithm with local adaptive selection methods of the bandwidth, $h_n(y)$. This is a open problem that will be solved in a future work. Finally, when $m$ is large, the algorithm is more sensitive to both the choice of the starting point and the selection of the $h_n$, see \cite[Section 4]{Rico-Lisandro-Luis}.  This is probably due to the fact that as the number of  parameters  increase so does  the complexity of the model. This is, clearly,  a consequence of what has been dubbed as  the curse of  dimensionality.  

\subsection{Discussion, limitations and perspectives}
\label{sec:discussion}

In this section, we address several pivotal practical aspects regarding the implementation and robustness of the proposed estimation procedure, specifically concerning the bandwidth selection, the dimensionality of the state space, and the estimation of asymptotic variance.

\paragraph{Bandwidth selection.}
As is the case for all kernel-based estimates, the choice of the bandwidth $h_n$ is crucial. In our study, we adopted a data-driven approach by scaling the bandwidth with the standard deviation estimate of the observation for each state. This choice ensures that condition S1 is satisfied while adapting to the scale of the data. While cross-validation (CV) methods are popular in density estimation, their computational cost in the context of the proposed iterative Robbins-Monro algorithm—which requires running the restoration step at each iteration—is prohibitive.
Furthermore, as observed in Figures 5 and 11, the variance of the regression estimates increases at the boundaries of the domain. This boundary effect is a well-known phenomenon in kernel regression with fixed bandwidths. Although local adaptive bandwidths could mitigate this issue and improve performance in sparse regions, the theoretical derivation of asymptotic normality under adaptive bandwidths for MS-NAR processes introduces significant technical challenges that fall beyond the scope of this paper. The results presented here confirm that our choice of $h_n$ yields valid confidence intervals in the regions with sufficient data density, validating the theoretical CLT. An alternative is to employ a data-driven procedure using kernel estimation with bandwidth selected via the Goldenshluger-Lepski approach (see \cite{Bertin2025}).

\paragraph{Influence of the state dimension $m$.}
The computational complexity of the proposed algorithm is largely dominated by the restoration step, which relies on the Forward-Backward algorithm. The cost of this step is of order $O(n \cdot m^2)$. Consequently, while the theoretical results hold for any finite $m$, the practical application to high-dimensional state spaces ($m \gg 1$) requires careful consideration. As $m$ increases, the risk of label switching and the time required for the Markov chain to explore the state space (mixing time) increase, potentially slowing down the convergence rate of the Robbins-Monro procedure. In our simulation with $m=3$, the algorithm converges efficiently. However, for significantly larger $m$, one might consider combining this approach with particle filters or variational approximations to alleviate the computational burden. To our knowledge, the theoretical study of the number of states in MS-NAR models has not been addressed. One alternative is to generalize the method proposed in \cite{Chaumaray2024} for non-parametric hidden Markov models, using the rank of an integral operator relying on the distribution of a pair of consecutive observations. In our case, this would involve using Proposition 2.3 in \cite{Rico-Lisandro-Luis}.

\paragraph{Estimation of the asymptotic covariance.}
It is important to note that the asymptotic covariance matrix $\Gamma$ (and consequently $\Sigma^*_n(y)$ in Theorem 3.2) depends on unknown quantities. In our numerical study, we replaced these theoretical values with their consistent empirical counterparts obtained from the final iteration of the algorithm (as detailed in Remark 3.4). By Slutsky's theorem, replacing the true asymptotic variance with a consistent estimator does not alter the asymptotic distribution of the standardized quantity. Thus, the confidence bands shown in Figures 5, 11, and 12 are asymptotically valid approximations. The simulation results support this theoretical claim, as the true functions largely remain within the estimated confidence regions.

\paragraph{Model complexity and robustness.}
Finally, regarding the simulation setup, we deliberately chose a model with $m=3$ regimes that shares structural similarities with previous studies but introduces specific complexities, such as intersecting regression functions (see Figure 3). These intersection points present a challenging scenario for identifiability and classification, providing a robust test bed for verifying the Asymptotic Normality results. By maintaining a comparable structure to previous works, we isolate the contribution of the Central Limit Theorem derived in this paper, specifically demonstrating the accuracy of the confidence intervals which were not available in prior consistency studies.

%%%%%%%%%%%%%%%%%%%%%%%%%%%%%%%%%%%%%%%%%%%%%%%%%%%%%%%%%%%%%%%%
%%%%%%%%%%%%%%%%%%%%%%%%%%%%%%%%%%%%%%%%%%%%%%%%%%%%%% Appendix 

\section{The proofs}\label{appA}

%%%%%%%%%%%%%% Proof of Lemma 3.2
{\bf Proof of Lemma \ref{asympt_cov2}.}
First, from H\"older inequality
\begin{equation}\label{ineq_lemma5}
\mathbb{E}\left(Y_1^2\car_{\{|Y_1|>M_n\}}K^2_h\left(y-Y_0\right)\right) 
\leq
\mathbb{E}(|Y_1|^{s}) \|Y_1^{2-s}\car_{\{|Y_1|>M_n\}}K^2 \|_{\infty} \\
\leq \mathbb{E}(|Y_1|^{s})\|K\|_{\infty}^2M_n^{2-s}.
\end{equation}
Thus, from conditions B1 and M2, we deduce that 
\begin{equation*}
\mathrm{Var}(R_{0,n})\leq  b^{2}\mathbb{E}(|Y_1|^{s})\|K\|^2_{\infty} M_n^{2-s}.
\end{equation*}
Here, $\mathbb{E}(|Y_1|^{s})< \infty$ $y$-a.s, from conditions E3 and M2.

On the other hand, from conditions E3, M2 and B1, by the Cauchy-Schwarz inequality and \eqref{ineq_lemma5},
we bound
\begin{equation*}
\mathrm{Cov}(R_{0,n},R_{k,n})
\leq \sqrt{\mathrm{Var}(R_{0,n})}\sqrt{\mathrm{Var}(R_{k,n})}
\leq b^{2}\mathbb{E}(|Y_1|^{s})\|K\|^2_{\infty} M_n^{2-s}.
\end{equation*}

Finally, in the similar way that in the proof of Lemma~\ref{asympt_cov}, from conditions 
E1, E6, D1, B2, R3 and inequality \eqref{exisdensidad_ii},  we have that
\begin{eqnarray*}
\mathrm{Cov}(R_{0,n},R_{k,n})
&\leq&
\!\!\!b^{2}
\mathbb{E}\left(Y_1K_h\left(y-Y_0\right)\car_{i}(X_1)Y_{k+1}K_h\left(y-Y_k\right)\car_{i}(X_{k+1})\right)\\
&=&\!\!\!b^2A^{(k)}_{i,i}\mu_i\int\!\!\!\!\int r_{i,k}(t,s)K_h\left(y-t\right)K_h\left(y-s\right)p(Y_0=t,Y_k=s)dtds\\
&=&\!\!\!h^2b^2A^{(k)}_{i,i}\!\mu_i\!\!\!\int\!\!\!\!\int\!\!\! r_{i,k}(y\!-\!ht,y\!-\!hs)K(t)K(s)p(Y_0\!=\!y\!-\!ht,Y_k\!=\!y\!-\!hs)dtds\\
&\preceq&\!\!\! h^2b^2r_{i,k}(y,y)A^{(k)}_{i,i}\mu_i\|\Phi\|_{\infty} + o(h^3).
\end{eqnarray*}
\cqd

%%%%%%%%%%%%%%%% Proof of Theorem 3.1 
{\bf Proof of Theorem \ref{normalidadTO}}.
We proceed as in \cite[Lemmas 4.3 and 4.4]{Ango-Rico-Dupoiron}. First, 
we show the following vectorial CLT for each $y$,
\begin{equation}
\label{TCLvectorial} 
\sqrt{nh_n}(\hat{f}_{i,n}(y)-f_i(y),\hat{g}_{i,n}(y)-g_i(y))\to\mathcal{N}(0,\Sigma_i(y)),
\end{equation}
where
\begin{equation*}
\Sigma_i(y)= \left(\begin{array}{cc}
               f_i(y) & g_i(y)\\
               g_i(y) & {g_2}_i(y)
              \end{array}\right)\|K\|_2^2.
\end{equation*}

We use the Cram\'er-Wold device: the vectorial CLT is equivalent to proof that for each $(a,b)\in\mathbb{R}^2$
\begin{equation}
\label{TCLCramer}
\sqrt{nh_n}(a(\hat{f}_{i,n}(y)-f_i(y))+b(\hat{g}_{i,n}(y)-g_i(y))\to\mathcal{N}(0,V_i(y)),
\end{equation}
where $n, h_n$ satisfy condition S1, and
$V_i(y)=\left(a^2f_i(y)+2abg_i(y)+b^2{g_2}_i(y)\right)\|K\|_2^2$.

Define, as in Lemma~\ref{asympt_cov}, the term $T_{k,n}$
%$$
%T_{k,n}=aK_{h_n}(y-Y_k)\car_i(X_{k+1})+bY_{k+1}\car_{\{|Y_{k+1}|\leq M_n\}}K_{h_n}(y-Y_k)\car_i(X_{k+1}),
%$$
where $\{M_n\}_{n\geq0}$ is a positive sequence converging to infinity, when $n\to \infty$.
Thus, the kernel estimator $\hat{l}_n(y) = a\hat{f}_{i,n}(y)+b\hat{g}_{i,n}(y)$ is splited in a truncated part
\begin{equation*}
\tilde{l}_n(y)=\frac{1}{nh_n}\sum_{k=0}^{n-1} T_{k,n},
\end{equation*}
and the remaining part of the truncation is
\begin{equation*}
\hat{l}_n(y)-\tilde{l}_n(y)=\frac{1}{nh_n} \sum_{k=0}^{n-1}R_{k,n}, 
\end{equation*}
with $R_{k,n}$ defined as in Lemma~\ref{asympt_cov2}.
%, i.e.
%$$
%R_{k,n}=b Y_{k+1}\car_{\{|Y_{k+1}|> M_n\}}K_{h_n}\left(y-Y_k\right)\car_i(X_{k+1}).
%$$

In order to establish the CLT in \eqref{TCLCramer}, the following steps will be verified, for each $y$, $h_n\to 0$ and $nh_n\to\infty$:
\begin{enumerate}
\item The asymptotic variance converges, 
$\tilde{V}_{i,n}(y)= nh_n\, \mathrm{Var}(\tilde{l}_n(y)) \to V_i(y)$.
\item The asymptotic normality of the truncated kernel estimator $\tilde{S}_n(y)=\sqrt{nh_n}\,(\tilde{l}_n(y) - \mathbb{E}(\tilde{l}_n(y)))$; i.e.  
$\tilde{S}_n(y)\to\mathcal{N}(0,V_i(y))$.
\item The estimator $\hat{S}_n(y)=\sqrt{nh_n}\,(\hat{l}_n(y) - \mathbb{E}(\hat{l}_n(y)))$ satisfies that
$\mathbb{E}\left( \hat{S}_n(y)-\tilde{S}_n(y)\right)^2\to0$.
\end{enumerate}

\noindent\textbf{Step 1}: Note that
$
\mathrm{Var}(\tilde{l}_n(y))=\frac{1}{nh_n^2}\mathrm{Var}(T_{0,k})
$
and by i) in Lemma \ref{asympt_cov}
\begin{equation*}
nh_n\mathrm{Var}(\tilde{l}_n(y))\approx (a^2f_i(y)\!+\!2abg_i(y)\!+\!b^2{g_2}_i(y))\|K\|_2^2\! + \! o(h_n),
\end{equation*}
then  for $nh_n\to \infty$ and $h_n\to 0$, 
$\tilde{V}_{i,n}(y)\to V_i(y)=(a^2f_i(x)+2abg_i(x)+b^2{g_2}_i(x))\|K\|_2^2$.

\noindent\textbf{Step 2}: In order to prove the CLT, we define the centered random variable $ D_k=T_ {k,n} -\mathbb{E}(T_ {k, n})$ and 
denote by $Q$ the associate quantile function. The sequence $\{D_k\}_{k\geq0}$  satisfies the following conditions:
\begin{itemize}
\item[i)] $\{D_k\}_{k\geq0}$ is a strictly stationary sequence and strongly $\alpha$-mixing.
 \item[ii)] In virtue of the assumptions  E3 and M2,  $\mathbb{E}(|D_k|^s)<\infty$, for some $s>2$.
 \item[iii)] For $\alpha^{-1}(u)=\inf\{k\in\mathbb{N}:\ \alpha_k\leq u \}$, we have
\begin{equation*}\label{DM}
\int_{0}^1\alpha^{-1}(u)Q^2(u)du<\infty. 
\end{equation*}
This condition is implied by
\begin{equation*}\label{condicionTCL}
\sum_{i\geq0}(i+1)^{\frac{2}{s-2}}\alpha_i<\infty,
\end{equation*}
and in our case this is valid, since from geometric $\alpha$-mixing property (see section \ref{strongmixing}) we have that there exist $c>0$ and $0<\zeta<1$ such that the mixing coefficients $\alpha_i\leq c \zeta^ i$, then
\begin{equation*}
\sum_{i\geq0}(i+1)^{\frac{2}{s-2}}\alpha_i\leq c\sum_{i\geq0}(i+1)^{\frac{2}{s-2}}\zeta^i<\infty.
\end{equation*}
\end{itemize}

Thus, applying Theorem 4.2. in \cite{Rio1}, we obtain the weak convergence of $\tilde{S}_n(y)$ to a $\mathcal{N}(0,V_i(y))$.

\noindent\textbf{Step 3}: Using Tran's truncation technique, for $1< u_n< n$, by Lemma \ref{asympt_cov2} we have
\begin{eqnarray*}
\nonumber
\mathbb{E}\left(\hat{S}_n-\tilde{S}_n\right)^2 &=& \frac{1}{nh_n}\mathrm{Var}\left(\sum_{k=0}^{n-1}R_{k,n}\right)\\ &=&\frac{1}{h_n}\mathrm{Var}(R_{0,n})+ \frac{2}{nh_n}\sum_{k=1}^{u_n}(n-k)\mathrm{Cov}(R_{0,n},R_{k,n})+ \frac{2}{nh_n}\sum_{k=u_n+1}^{n-1}(n-k)\mathrm{Cov}(R_{0,n},R_{k,n})\\
&\preceq & b^{2}\mathbb{E}(|Y_1|^{s})\|K \|_{\infty}^2M_n^{2-s}h_n^{-1} +
2u_n h_n\left[b^2r_{i,k}(y,y)\|\Phi\|_{\infty}+o(h_n)\right]+ 2b^{2}\mathbb{E}(|Y_1|^{s})\|K \|_{\infty}^2 nh_n^{-1} M_n^{2-s}.
\end{eqnarray*}
Taking $u_n=1/h_n \log n$, $h_n=n^{-d}$ and $M_n=M_0n^{\gamma}$, with $0<d<1$ and $\gamma >0$ such that  $(1+d)+(2-s)\gamma < 0$, then 
\begin{equation*}
\mathbb{E}\left(\hat{S}_n-\tilde{S}_n\right)^2 = O(1/\log(n)).
\end{equation*}

This implies that $\hat{S}_n(y)\to\mathcal{N}(0,V_i(y))$.
Moreover, from Lemma 2.4 in \cite{Rico-Lisandro-Luis} we have 
\begin{equation*}
\sup_{y\in\mathrm{C}}|\mathbb{E}\hat{g}_{i,n}(y)-g_i(y)|=O(h_n^2),\quad \mbox{ and } \quad \sup_{y\in\mathrm{C}}|\mathbb{E}\hat{f}_{i,n}(y)-f_i(y)|=O(h_n^2).
\end{equation*}
Therefore, for $d>1/5$, the CLT in \eqref{TCLCramer} is proved.

Finally, we use the fact that $a^2f_i(y)+2abg_i(y)+b^2{g_2}_i(y)$ can be written 
as a quadratic form
\begin{equation*}
(a^2f_i(y)+2abg_i(y)+b^2{g_2}_i(y))\!=\!(a,b)\left(\!\!\begin{array}{cc}
                                          f_i(y) & g_i(y)\\
                                          g_i(y) & {g_2}_i(y)
                                         \end{array}\!\!\right)\left(\!\!\begin{array}{c}
                                                            a\\
                                                            b
                                                           \end{array}\!\!\right).
\end{equation*}
Thus, the Cram\'er-Wold device implies that 
\begin{equation*}
\sqrt{nh_n}(\hat{f}_{i,n}(y)-f_i(y),\hat{g}_{i,n}(y)-g(y))\to\mathcal{N}(0,\Sigma_i(y)) 
\end{equation*}

On the other hand, Coulomb's decomposition gives
\begin{equation*}
\hat{r}_{i,n}-r_i=\frac{\hat{g}_{i,n}f_i-\hat{f}_{i,n}g_i}{f_i^2}+\left(\frac{1}{\hat{f}_{i,n}}-\frac{1}{f_i}\right)(\hat{g}_{i,n}-g_i) -\left(\frac{r_i}{\hat{f}_{i,n}}-\frac{g_i}{f_i^2}\right)(\hat{f}_{i,n}-f_i).
\end{equation*}
Then, setting  $a=-g_i/f_i^2$, $b=1/f_i$, since 
\begin{eqnarray*}
\hat{g}_{i,n}f_i-\hat{f}_{i,n}g_i= (\hat{g}_{i,n}-g_i)f_i-(\hat{f}_{i,n}-f_i)g_i,
\end{eqnarray*}
we deduce from the Cram\'er-Wold device \eqref{TCLCramer} that 
\begin{eqnarray*}
\sqrt{nh_n}\, \frac{\hat{g}_{i,n}(y)f_i(y)-\hat{f}_{i,n}(y)g_i(y)}{f_i^2(y)}\to \mathcal{N}\left(0,\frac{\sigma_i^2(y)\|K\|^2_2}{f_i(y)}\right).
\end{eqnarray*}

The vectorial CLT  \eqref{TCLvectorial} implies
\begin{equation*}
\hat{f}_{i,n}(y)=f_i(y)+O_{\mathbb{P}}\left(\frac{1}{\sqrt{nh_n}}\right), \;\; \mbox{ and } \;\; \hat{g}_{i,n}(y)=g_i(y)+O_{\mathbb{P}}\left(\frac{1}{\sqrt{nh_n}}\right).
\end{equation*}

According to Lemma 2.4 in \cite{Rico-Lisandro-Luis} we have
\begin{equation}\label{fpositive}
|\hat{f}_{i,n}(y)|\geq |f_i(y)| - |\hat{f}_{i,n}(y)-\mathbb{E}\hat{f}_{i,n}(y)|
-|\mathbb{E}\hat{f}_{i,n}(y)-f_i(y)|
\geq \frac{1}{2}|f_i(y)| >0.
\end{equation}
Moreover, following \cite[Lemmas 2.3 and 2.4]{Rico-Lisandro-Luis} we have $\hat{f}_{i,n}(y)-f_i(y) \to 0, \, a.s.$ Then, from inequality \eqref{fpositive} we have
\begin{equation*}
\frac{1}{\hat{f}_{i,n}(y)}-\frac{1}{f_i(y)} \to 0, \quad a.s.
\end{equation*}
Therefore, 
\begin{equation*}
\sqrt{nh_n}\left(\frac{1}{\hat{f}_{i,n}(y)}-\frac{1}{f_i(y)}\right)(\hat{g}_{i,n}(y)-g_i(y))\to0\quad \mbox{in probability},
\end{equation*}
and 
\begin{equation*}
\sqrt{nh_n}\left(\frac{r_i(y)}{\hat{f}_{i,n}(y)}-\frac{g_i(y)}{f_i^2(y)}\right)(\hat{f}_{i,n}(y)-f_i(y))\to0\quad \mbox{in probability},
\end{equation*}
since $g_i=r_if_i$. Thus, the proof is complete. 
\cqd

%%%%%%%%%%%%%%%% Proof of Lemma 3.3 
{\bf Proof of Lemma \ref{cond_martingale}}. In order to obtain the condition $i)$ is sufficient to prove that $\varsigma_t$ has 
a conditional  moment of order $s>2$, in fact by H\"older inequality
\begin{equation}\label{inq_interp}
\mathbb{E}(\|\varsigma_{t}\|^2 \car(\|\varsigma_{t}\| > c)|\mathcal{F}_{t-1}) \leq
\frac{\mathbb{E}(\|\varsigma_{t}\|^s|\mathcal{F}_{t-1})}{c^{s-2}}.
\end{equation}

Following the proof of \cite[Theorem 3.1]{Rico-Lisandro-Luis} we obtaing a.s. the following inequality for the conditional second moment of $\|\varsigma_{t}\|$ 
\begin{equation*}
\mathbb{E}(\|\varsigma_{t}\|^2|\mathcal{F}_{t-1})
 \leq \| \Psi(\theta^{t-1})\|^2,
\end{equation*}
where $\Psi(\theta)=(\Psi_1(\theta),\ldots,\Psi_m(\theta))$ and
\begin{equation*}
\Psi_i(\theta)= \frac{1}{nh}\sum_{k=0}^{n-1}(Y_{k+1}-\theta_i)K_h\left(y-Y_k\right).
\end{equation*}
Analogously, we can obtain a similar inequality for the moment of order 4. Thus,  using an argument of interpolation between $L^2$ and $L^4$ spaces we obtain, for $2\leq s\leq4$,
\begin{equation}\label{inq_unif}
\mathbb{E}(\|\varsigma_{t}\|^s|\mathcal{F}_{t-1})\leq J_s \| \Psi(\theta^{t-1})\|^s,
\end{equation}
where $J_s$ is a positive constant. 

Since the closure of $\{\theta^t\}$ is a compact subset of $\Theta$ then $\|\Psi\|_{\infty}<\infty$ and
\begin{equation*}
\limsup_{t \to\infty} \mathbb{E}(\|\varsigma_{t}\|^s|\mathcal{F}_{t-1})\leq J_s \|\Psi\|_{\infty}^s.
\end{equation*}
Thus, from \eqref{inq_interp} and \eqref{inq_unif} condition $i)$ is verified.

On the other hand, from definition of $\varsigma_t$, the $(i,j)-$th element of the matrix $\Gamma^t=\mathbb{E}[\varsigma_t\varsigma_t^T|\mathcal{F}_{t-1}]$ is
\begin{equation*} 
\Gamma^t_{i,j}\!=\!\frac{4}{n^2h^2}\!\!\!\sum_{k,k'=0}^{n-1}\!\!(Y_{k+1}\!-\!\theta_i^{t-1})\!(Y_{k'+1}
\!-\!\theta_j^{t-1})K_h(y\!-\!Y_k)K_h(y\!-\!Y_{k'})\chi^t_{i,j}(k,k'),
\end{equation*}
where $
\chi^t_{i,j}(k,k')=\mathrm{Cov}(\!\car_i(X_{k+1}^{t}\!),\!\car_j(X_{k'+1}^{t}\!)|\mathcal{F}_{t-1})$.

Finally, from Remark \ref{covBernoulli} and the a.s. convergence of $\theta^{t}$ to $\theta^{*}$ given in  \cite[Theorem 3.1]{Rico-Lisandro-Luis}  we have the convergence in probability of $\Gamma_{i,j}^t$ to $\Gamma_{i,j}$ when $t\to\infty$, with
\begin{equation*} 
\Gamma_{i,j}=\frac{4}{n^2h^2}\!\sum_{k,k'=0}^{n-1}\!(Y_{k+1}-\theta_i^{*})\!(Y_{k'+1}
-\theta_j^{*})K_h(y-Y_k)K_h(y-Y_{k'})\chi_{i,j}(k,k').
\end{equation*} 
\cqd

%%%%%%%%%%%%%%%% Proof of Theorem 3.2
{\bf Proof of Theorem \ref{tasadeconvergencia}}.We denote $H=\nabla_\theta^2 u(y,Y_{0:n},\theta^*)$.
Consider the sequence $\{M_t\}_{t\geq0}$ defined by 
\begin{equation*}\label{rmlineal}
M_t=\left(I-\gamma_t H\right)M_{t-1}+\gamma_t\varsigma_t,
\end{equation*}
with $M_0=0$. This suggests that the error can be decompose in the form $\Delta_t=M_t+R_t$, where 
\begin{equation*}
R_t=\prod_{k=1}^{t}\left(I-\gamma_k H \right)\Delta_0.
\end{equation*}

Thus, $\overline{\Delta}_t= \overline{M}_{t}+\overline{R}_t$ where $\overline{M}_{t}=t^{-1}\sum_{k=1}^tM_{k}$ and $\overline{R}_t=t^{-1}\sum_{k=1}^tR_k$.
By applying Lemmas 1 and 2 of \cite{Polyakuditsky}, we can rewrite
\begin{equation*}
{\sqrt{t}}\;\overline{\Delta}_{t}=\frac{1}{\sqrt{t}\gamma_0}\nu_t {\Delta}_{0}+
\frac{1}{\sqrt{t}}\sum_{k=1}^{t} H^{-1}\varsigma_k+
\frac{1}{\sqrt{t}}\sum_{k=1}^{t}\omega_k^t\varsigma_k,
\end{equation*}
where 
\begin{equation*}
\omega_k^t=\nu_k^t-H^{-1}, \quad \nu_k^t=\gamma_k \sum_{i=k}^{t}\prod_{j=k+1}^i(I-\gamma_j H), \quad \nu_t=\nu_0^t,
\end{equation*} 
and setting $\prod_{j=k+1}^k(I-\gamma_j H)=I$. 

The matrix $\nu_t$ and $\omega_k^t$ are such that 
$\|\nu_t\|\leq c$, $\|\omega_k^t\|\leq c$ for some  $c<\infty$, and  
$\lim_{t\to\infty}\frac{1}{t}\sum_{k=1}^{t}\|\omega_k^t\|=0
$. Furthermore, as $\|\nu_t\|\leq c$, it also holds 
\begin{equation*}
\lim_{t\to\infty}\frac{1}{\sqrt{t}\gamma_0}\nu_t {\Delta}_{0}=0.
\end{equation*}

On the other hand, from condition \eqref{M2Marting} we has  
\begin{equation*}
\lim_{t\to\infty}\mathbb{E}\left(\frac{1}{\sqrt{t}}\sum_{k=1}^{t}\omega_k^t \varsigma_k\right)^2
\leq c \lim_{t\to\infty}\frac{1}{t}\sum_{k=1}^{t}\|\omega_k^t\|=0.
\end{equation*}
Thus, the weak convergence to the normal distribution is dominated by the term
\begin{equation*}
\frac{1}{\sqrt{t}}\sum_{k=1}^{t} H^{-1}\varsigma_k.
\end{equation*}

Finally, from Lemma \ref{cond_martingale} and Cesaro's Theorem we can infer the following conditions for the Martingales Central Limit Theorem:
\begin{itemize}
 \item[i)]
$\limsup\limits_{t\to\infty}\frac{1}{t}
\sum\limits_{k=1}^{t}\mathbb{E}\left(\|H^{-1}\varsigma_k\|^2
\car\left(|H^{-1}\varsigma_k|> c\right)|\mathcal{F}_{t-1}\right)=0$,
in probability.
\item[ii)]
$\lim\limits_{t\to\infty}\frac{1}{t} \sum\limits_{k=1}^{t}
H^{-1}\mathbb{E}\left(\varsigma_k\varsigma_k^{T}|\mathcal{F}_{k-1}
\right)H^{-1}= H^{-1}\Gamma H^{-1}=\Sigma^*_n(y)$, in probability.
\end{itemize}
Thus, the weak convergence is obtained. \cqd

\bibliographystyle{plain}
\bibliography{arhmm}

\begin{thebibliography}{10}

\bibitem{Alexandrovich}
G.~Alexandrovich, H.~Holzmann, and A.~Leister.
\newblock Nonparametric identification and maximum likelihood estimation for
  hidden markov models.
\newblock {\em Biometrika}, 103(2):423--434, 03 2016.

\bibitem{Allman}
E.S. Allman, C.~Matias, and J.A. Rhodes.
\newblock {Identifiability of parameters in latent structure models with many
  observed variables}.
\newblock {\em Ann. Statist.}, 37:3099--3132, 2009.

\bibitem{Ango-Rico-Dupoiron}
P.~Ango-Nze, S.~Dupoiron, and R.~R\'\i os.
\newblock Subsampling under weak dependence conditions.
\newblock Working paper, Universit{\'e} de Lille 3 and Universidad Central de
  Venezuela, 2003.
\newblock Available at ResearchGate.

\bibitem{Balsells}
C.~Balsells-Rodas, Y.~Wang, and Y.~Li.
\newblock On the identifiability of switching dynamical systems.
\newblock In {\em In Forty-first International Conference on Machine Learning},
  pages 52639--2672, July 2024.

\bibitem{Bertin2025}
K.~Bertin, L.~Fermin, and M.~Padrino.
\newblock Adaptive estimation in regression models for weakly dependent data
  and explanatory variable with known density.
\newblock {\em Statistics}, pages 1--11, 2025.

\bibitem{cappe-moulines-ryden}
O.~Capp\'e, E.~Moulines, and T.~Ryd\'en.
\newblock {\em {\it Inference in Hidden Markov Models }}.
\newblock Springer-Verlag, 2005.

\bibitem{Chau}
T.~Chau, P.~Ailliot, and V.~Monbet.
\newblock An algorithm for non-parametric estimation in state--space models.
\newblock {\em Computational Statistics \& Data Analysis}, 153, 2021.

\bibitem{Chaumaray2024}
M.~Du~Roy de~Chaumaray, S.~El Koleim, M-P. Etienne, and M.~Matthieu.
\newblock Estimation of the order of non-parametric hidden {Markov} models
  using the singular values of an integral operator.
\newblock {\em Journal of Machine Learning Research}, 25(415):1--37, 2024.

\bibitem{douc}
R.~Douc, E.~Moulines, and T.~Ryd{\'e}n.
\newblock {Asymptotic properties of the maximum likelihood estimator in
  autoregressive models with Markov regime}.
\newblock {\em Ann. Statist.}, 32:2254--2304, 2004.

\bibitem{dmixing}
P.~Doukhan.
\newblock {\em {\it Mixing: Propierties and Examples.}}, volume~85.
\newblock Lecture Notes in Statist., 1994.

\bibitem{duflo}
M.~Duflo.
\newblock {\em {\it Algorithmes Stochastiques}}.
\newblock Springer-Verlag, Berlin, 1996.

\bibitem{gassiat5}
L.~Lehéricy E.~Gassiat, S. Le~Corff.
\newblock {Identifiability and Consistent Estimation of Nonparametric
  Translation Hidden Markov Models with General State Space}.
\newblock {\em The Journal of Machine Learning Research}, 21(1):4589 -- 4628,
  2020.

\bibitem{Marcano-Lisandro-Luis}
L.~Ferm\'in, J.~Marcano, and L.A. Rodr\'iguez.
\newblock {A proof of consistency of the MLE for nonlinear Markov-switching AR
  processes}.
\newblock {\em Statistics and Probability Letters}, 183:1--7, 2021.

\bibitem{Rico-Lisandro-Luis}
L.~Ferm\'in, R.~R\'ios, and L.A. Rodr\'iguez.
\newblock {A Robbins Monro algorithm for nonparametric estimation of NAR
  process with Markov-Switching: consistency}.
\newblock {\em Journal of Time Series Analysis}, 38(6):809--837, 2017.

\bibitem{gassiat4}
E.~Gassiat and J.~Rousseau.
\newblock {Nonparametric finite translation hidden Markov models and
  extensions}.
\newblock {\em Bernoulli}, 22(1):193–212, 2016.

\bibitem{Goldfeld}
S.~M. Goldfeld and R.~Quandt.
\newblock {A Markov Model for Switching Regressions}.
\newblock {\em Journal of Econometrics}, 1:3--16, 1973.

\bibitem{Hall-Zhou}
P.~Hall and X.-H. Zhou.
\newblock {Nonparametric estimation of component distributions in a
  multivariate mixture}.
\newblock {\em Ann. Stat.}, 31(1):201–224, 2003.

\bibitem{Hamilton}
J.D. Hamilton.
\newblock {A new approach to the economic analysis of non stationary time
  series and the business cycle}.
\newblock {\em Econometrica}, pages 357--384, 1989.

\bibitem{Harel-Puri}
M.~Harel and M.L. Puri.
\newblock {Universally consistent conditional U-statistics for absolutely
  processes and its applications for hidden Markov models.}
\newblock {\em Annals of the Institute of Statistical Mathematics}, 56
  (4):819--832, 2004.

\bibitem{Kruskal}
J.~B. Kruskal.
\newblock {Three-way arrays: rank and uniqueness of trilinear decompositions,
  with application to arithmetic complexity and statistics}.
\newblock {\em Linear Algebra and its Applications}, 18(2):95--138, 1977.

\bibitem{Nademi}
A.~Nademi and Y.~Nademi.
\newblock Forecasting crude oil prices by a semiparametric {Markov} switching
  model: {OPEC}, {WTI}, and {Brent} cases.
\newblock {\em Energy Economics}, 74:757--766, August 2018.

\bibitem{Polyakuditsky}
B.~T. Polyak and A.~B. Juditsky.
\newblock {Acceleration of stochastic approximation by averaging}.
\newblock {\em SIAM J. Control Optim.}, 30:838--855, 1992.

\bibitem{Rio1}
E.~Rio.
\newblock {\em {\it Th\'eorie asymptotique des processus faiblement
  d\'ependents}}, volume~31.
\newblock Springer-SMAI: Paris., 2000.

\bibitem{luis1}
R.~R\'{\i}os and L.~A. Rodr\'{\i}guez.
\newblock Estimaci\'on semiparam\'etrica en procesos autorregresivos con
  r\'egimen de markov.
\newblock {\em Divulgaciones Matem\'aticas}, 16(1):155--171, 2008.

\bibitem{YaoRe}
J.~Yao.
\newblock {On Recursive Estimation in Incomplete Data Models}.
\newblock {\em Statistics}, 34:27--51, 2000.

\bibitem{Yao}
J.~Yao and J.~G. Attali.
\newblock {On stability of nonlinear AR process with Markov switching}.
\newblock {\em Adv. Applied Probab}, 32:394--407, 1999.

\end{thebibliography}

\end{document}